\documentclass[10pt,journal,comsoc]{IEEEtran}
\pdfoutput=1
\ifCLASSOPTIONcompsoc
  \usepackage[nocompress]{cite}
\else
  \usepackage{cite}
\fi
\ifCLASSINFOpdf
\else
\fi
\usepackage{textcomp}
\usepackage{xcolor}
\usepackage{listings}
\usepackage{indentfirst}
\usepackage{amsfonts}
\usepackage{algorithm}
\usepackage{algorithmic}
\usepackage{graphicx}
\usepackage{amsmath}
\usepackage{amsthm}

\newtheorem{theorem}{Theorem}
\newtheorem{corollary}{Corollary}

\newtheorem{remark}{Remark}
\newtheorem{proposition}{Proposition}
\newtheorem{assumption}{Assumption}

\usepackage{subfigure} 
\usepackage[section]{placeins}
\usepackage{float}
\usepackage{multirow}
\usepackage{balance}
\usepackage{color}
\usepackage{ragged2e}
\usepackage{array}

\def\BibTeX{{\rm B\kern-.05em{\sc i\kern-.025em b}\kern-.08em
    T\kern-.1667em\lower.7ex\hbox{E}\kern-.125emX}}
    
\begin{document}

\newcommand*{\Scale}[2][4]{\scalebox{#1}{$#2$}}%

%
\title{Online Pricing Incentive to Sample Fresh Information}
%
%
%
%

\author{Hongbo~Li,
        and~Lingjie~Duan,~\IEEEmembership{Senior~Member,~IEEE}
\thanks{Hongbo Li and Lingjie Duan are with the Pillar of Engineering Systems and Design, Singapore University of Technology and Design, Singapore 487372 (e-mail: hongbo\_li@mymail.sutd.edu.sg; lingjie\_duan@sutd.edu.sg).}}

\IEEEtitleabstractindextext{%
\begin{abstract}
\justifying
Today mobile users such as drivers are invited by content providers (e.g., Tripadvisor) to sample fresh information of diverse paths to control the age of information (AoI). However, selfish drivers prefer to travel through the shortest path instead of the others with extra costs in time and gas. To motivate drivers to route and sample diverse paths, this paper is the first to propose online pricing for a provider to economically reward drivers for diverse routing and control the actual AoI dynamics over time and spatial path domains. This online pricing optimization problem should be solved without knowing drivers' costs and even arrivals, and is intractable due to the curse of dimensionality in both time and space. If there is only one non-shortest path, we leverage the Markov decision process (MDP) techniques to analyze the problem. Accordingly, we design a linear-time algorithm for returning optimal online pricing, where a higher pricing reward is needed for a larger AoI. If there are a number of non-shortest paths, we prove that pricing one path at a time is optimal, yet it is not optimal to choose the path with the largest current AoI. Then we propose a new backward-clustered computation method and develop an approximation algorithm to alternate different paths to price over time. Perhaps surprisingly, our analysis of approximation ratio suggests that our algorithm's performance approaches closer to optimum given more paths.
\end{abstract}

}

\maketitle

\IEEEdisplaynontitleabstractindextext

%
\IEEEpeerreviewmaketitle

\section{Introduction}\label{section1}

%
%
%
%
\IEEEPARstart{T}{oday} content providers (e.g., Tripadvisor, Yelp and Google Maps) prefer not to deploy expensive dedicated sensor networks to cover the whole city or nation. Instead, they invite mobile users such as drivers to sample fresh information on diverse paths especially those infrequently visited in the past \cite{wang2016sparse}. The sampled live information on the way include air quality data, shopping promotion and location, and traffic condition \cite{wang2016sparse,asghar2014review,bhoraskar2012wolverine}. However, selfish drivers are not willing to travel through non-shortest paths to sample fresh information if there are no rewards to offset their extra costs in time and gas \cite{li2019recommending}. It is shown in \cite{figueiredo2004exploiting} that the network performance degrades significantly due to drivers' selfish routing behaviors. To leverage the power of the mobile crowd, it is critical for providers to properly offer monetary rewards to drivers to change their myopic routing decisions for sampling fresh information along diverse paths. Such incentive mechanisms should be designed in an online version to adapt to the actual variations of information freshness over paths and time. 

Recent crowdsourcing works study optimal sensing policies by recruiting mobile vehicles to sense data (e.g., \cite{xiao2017mobile,he2015high,xu2019ilocus,lai2022optimized}). For example, \cite{xiao2017mobile} takes advantage of the mobility of vehicles to provide location-based services in large scale areas. \cite{he2015high} defined spatial and temporal coverage as two metrics for crowdsourcing quality to design greedy and genetic approximation algorithm. \cite{xu2019ilocus} incentivizes a crow of vehicle drivers to sense and sample the desired target regions in one-shot. \cite{lai2022optimized} allows allocated vehicles to follow their origin and destination routes while maximize the overall sensing benefit. However, all of these studies overlook the freshness of sampled data.

To model the information freshness, \cite{kaul2012real} proposes the concept of age-of-information (AoI) before a new information update is received. Following this, both time-average and peak AoI are developed to measure the average and maximum ages in the time domain, respectively \cite{costa2014age,kosta2017age}. Numerous works have been analyzing AoI statistics and designing scheduling policies to minimize these two performance metrics  \cite{hsu2019scheduling,sun2018age,kadota2018scheduling,li2019minimizing}. For example, by formulating the AoI state updating as a Markov decision process (MDP), \cite{hsu2019scheduling} develops approximation algorithms to schedule information packets from multi-sources to end-users for minimizing the average AoI. \cite{sun2018age} proposes online scheduling policies to minimize the AoI under different metrics in multi-flow, multi-server systems. \cite{kadota2018scheduling} compares several scheduling policies to find that transmitting the packet with the largest current AoI is optimal. \cite{he2016optimal} studies link scheduling to optimize of min-max peak AoI in wireless networks. However, these works only consider the passive packet arrival process for transmission protocols design, and overlook the opportunity to design the sampling process. 

To actively sample fresh data, several works study the path planning problem in UAV-assited IoT to minimize AoI \cite{zhou2019deep,jia2019age,li2019minimizing_uav}. \cite{zhou2019deep} formulates an MDP to capture the dynamics of UAV locations and applies reinforcement learning to solve the problem. \cite{jia2019age} jointly considers energy consumption and AoI evolution to study the data acquisition problem.   However, these works only consider the fully controlled UAV to route but overlook vehicle users' selfish behaviors to disobey. As today many mobile users are invited to sample information, it is important to leverage human-involved sampling or crowdsoucing to control AoI for content providers. 

Only a few mobile crowdsourcing works have taken the economics effect of AoI into consideration for information update or reward maximization. \cite{altman2019forever} considers sampling costs for users to decide when to self-update local data, without considering any incentive design from the content providers. \cite{zhang2021pricing} studies how an information customer requests and pays for the fresh information to be updated by the source. \cite{he2021optimal} studies a two-stage game model for a fresh data market to maximize a platform's profit, where the platform provides data with different AoI to dynamically arriving users. \cite{modina2022joint} jointly designs optimal upload strategy and incentivizes users to offload data in order to control the AoI. However, these works do not study how to motivate the power of the crowd for sampling fresh data to control AoI. The most related work to this paper is \cite{wang2019dynamic}, which proposes an offline algorithm to provide sampling rewards to mobile users for controlling expected AoI. Yet this work cannot adapt to unexpected AoI change for dynamic incentive design, and it only looks at a single path instead of a road network for information sampling.

To our best knowledge, this paper is the first work to study how a content provider designs its optimal pricing to reward drivers online for diverse routing and fresh information sampling. To best adapt to the actual variations of AoI over paths and time, we formulate our online optimization problem as a stochastic dynamic program. However, we need to overcome the following technical challenges.
\begin{itemize}
    \item \emph{Incomplete information about drivers' cost and arrival pattern}: To control the actual AoI evolution in real time, our incentive pricing as compensation should be designed according to drivers' actual arrivals and extra travel costs on non-shortest paths. Yet in practice, such information are private to drivers and unavailable for the content provider when deciding the online pricing. This makes it infeasible to apply online control methods such as Hamilton-Jacobi-Bellman equations and neural networks approximation here \cite{al2008discrete,dierks2010optimal}.
    \item \emph{Curse of dimensionality in both time and spatial path domains}: The optimal online pricing should be time- and path-dependent, making the number of system states exponentially increase with the time horizon $T$ and path number $N$ in backward induction \cite{puterman2014markov,littman2013complexity,alsheikh2015markov}. Some recent work finds special feature of system states (e.g., periodicity) to simplify the backward/forward induction yet cannot apply to our AoI problem (e.g., \cite{zhou2019collaborative,xiong2020reward}).
\end{itemize}

The existing AoI works design offline scheduling policies for multi-channel network by dynamic programming (e.g., \cite{wang2019dynamic,liu2019minimizing}). There are some recent work on online scheduling, yet they assume the system space to be countable and linearly increasing with number of channels and time (e.g., \cite{sun2018age,kadota2018scheduling,bedewy2019age}). MDP techniques are widely used to model the dynamic pricing problem (e.g., \cite{fang2020dynamic,rambha2016dynamic,baktayan2022intelligent}). \cite{fang2020dynamic} discretizes the state space into the set of finite intervals to greatly reduce its size, which is not applicable in our multi-path problem with possibly unbounded AoI. \cite{rambha2016dynamic} proposes to aggregate similar states to reduce the computational times, but it cannot provide any rigorous performance analysis such as proving approximation error in the worst case. \cite{baktayan2022intelligent} applies heuristic methods to solve the MDP under completed cost information, which did not analyze the error bound or prove structural properties. Their algorithm design and performance analysis cannot apply to our online pricing problem to sample fresh information in a large-scale road network over time.  

Our paper aims to overcome the above technical challenges for online pricing design, and our key novelty and main contributions in this paper are summarized as follows.
\begin{itemize}
    \item \emph{Novel online pricing to sample fresh information over time and space:} To motivate the power of the mobile crowd, this paper is the first to propose online pricing incentive for a content provider to economically reward drivers to sample fresh information through diverse routing. We optimize pricing reward on human-involved sampling to affect the actual AoI dynamics over time and spatial paths, without knowing drivers' hidden extra travel costs and even arrivals.
    \item \emph{Linear complexity algorithm to solve online pricing for one non-shortest path:} To minimize the state space of the our formulated MDP problem, we exploit the dynamics of AoI feature to jointly apply a fixed look-up table to greatly simplify the backward induction process. Then we design a linear-time algorithm for returning optimal online pricing. We prove structural properties of our unique pricing solution, and show that a higher pricing reward is needed for a larger AoI. We also show that our algorithm is applicable to infinite time horizon.  
    \item \emph{Approximation algorithm to solve online pricing for multiple non-shortest paths:} We propose a new backward-clustered computation method to overcome the curse of dimensionality in the path number. We first prove that it is optimal to only price one path at a time, while it is not optimal to myopically choose the path with the largest current AoI. Based on the backward-clustered method, we develop a new approximation algorithm to alternate different paths to price over time. This algorithm has only polynomial-time complexity, and its complexity does not depend on the number of paths. Perhaps surprisingly, our analysis of approximation ratio suggests that this algorithm's performance approaches closer to the optimum if more paths are involved to sample.
\end{itemize}

The rest of the paper is organized as follows. In Section \ref{section2},
we overview the system model to sample fresh information and introduce the problem formulation for online pricing. In Sections \ref{section3} and \ref{section4}, we prove structural properties of our unique pricing solution, and then propose a linear-time algorithm to return optimal online pricing if there is only one non-shortest path. In Section \ref{section5}, we develop a new approximation algorithm to alternate different paths to price over time if there are a number of non-shortest paths, which has only polynomial-time complexity and does not depend on the path number. Finally, we conclude this paper in Section \ref{section6}. 

\section{System Model and Problem Formulation}\label{section2}

As illustrated in Fig. \ref{Single_Path}, a random flow of drivers sequentially arrive at the gateway X in the discrete time horizon of in total $T$ rounds. At any time slot $t = 0,\cdots, T$, a driver (if any) needs to make routing decision from X to destination Y. Here, the road network includes the shortest path with normalized zero travel cost and another distant path with extra travel delay $D$ in the upper part of Fig. \ref{Single_Path}. For ease of exposition, we first focus on this case with only one non-shortest path, and then extend our analysis and pricing design to an arbitrary number of non-shortest paths later in Section \ref{section5}. 

\begin{figure*}[htbp]
    \centering
    \includegraphics[width=0.9\textwidth]{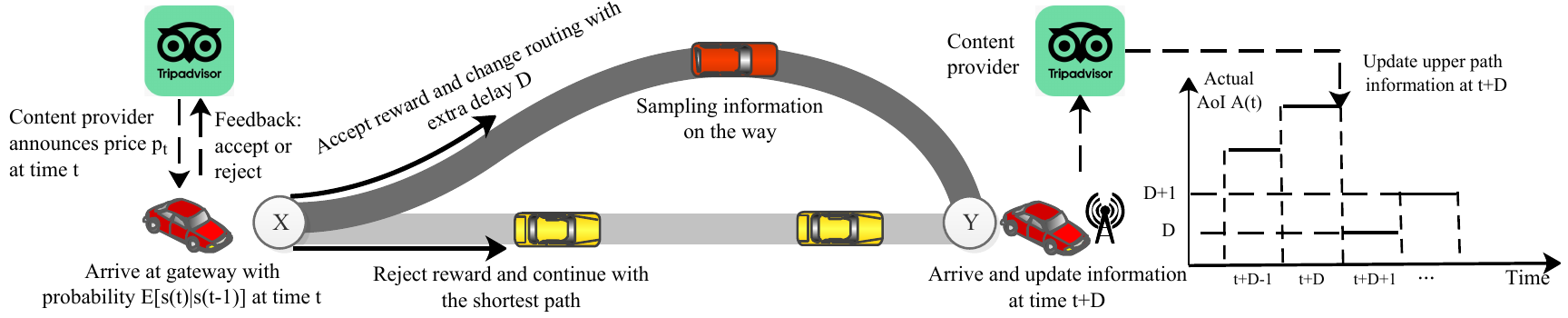}
    \caption{System model of pricing-driven sampling, where a random flow of drivers sequentially arrive at the gateway X and make routing decisions from X to destination Y. The road network includes the shortest path with normalized zero travel cost and another distant upper path with extra travel delay $D$.}
    \label{Single_Path}
\end{figure*}

Selfish drivers hesitate to choose the non-shortest upper path to incur extra travel cost in time and gas, resulting in under-sampling of this path (e.g., to collect air quality data, shopping promotion and location, or traffic condition \cite{wang2016sparse,asghar2014review,bhoraskar2012wolverine}). There are always enough drivers to cover the shortest path and we do not need to consider the AoI there. Our problem is how to motivate the randomly arriving drivers to sample the non-shortest path to control AoI there. At each time $t$, the content provider observes the actual AoI $A(t)$ of the non-shortest path (see Fig. \ref{Single_Path}), and compensates an arrived driver at the gateway with pricing reward $p_t$ to possibly change his routing to the upper path to return the sampled information of the whole path at time $t+D$. 

At each time slot $t$, the content provider can predict the ongoing AoI evolution till $t+D$, and its pricing decision needs to be adaptive to foreseeable AoI evolution in set \[\mathcal{A}(t, t+D)=\{A(\tau) |t\leq \tau\leq t+D\}.\] We can thus rewrite price $p_t$ at time $t$ as $p_t(\mathcal{A}(t, t+D))$ or simply $p_t(A(t+D))$. Note that the content provider will not decide any price after time $T-D$, as a driver can no longer return the sampled information timely before the end time $T$. 

Next we first introduce the driver's model and then present the online pricing problem formulation.

\subsection{Driver's arrival and cost model for sampling}
Following the traffic control literature (e.g., \cite{li2019recommending,yu2003short}), we model drivers' random arrivals at the gateway X over time as a Markov chain. As shown in Fig. \ref{event}, if a driver arrives at the beginning of time slot $t$, we denote it as $s(t)=1$, and otherwise $s(t)=0$. Each time slot's duration is set small enough such that there is at most one arrival at a time, and we practically model the correlation between arrivals across neighboring time slots. That is, if $s(t)=1$, 
\begin{equation}\notag
    s(t+1)=\begin{cases}
    1, &\text{with probability }1-\beta,\\
    0, &\text{with probability } \beta.
    \end{cases}
\end{equation}
Similarly, if $s(t)=0$, we update $s(t+1)$ by replacing the two probabilities $1-\beta$ and $\beta$ above by $\alpha$ and $1-\alpha$.
At the beginning of time slot $t$, the content provider only knows past arrival information $s(t-1)$ in Fig. \ref{event}, and expects arrival probability during time slot $t$ by using the Markov chain, i.e.,
\begin{equation}
    \mathbb{E}\big[s(t)|s(t-1)\big]=s(t-1)(1-\beta)+\big(1-s(t-1)\big)\alpha.\label{Est}
\end{equation}

\begin{figure}[htbp]
\centering
\includegraphics[width=0.36\textwidth]{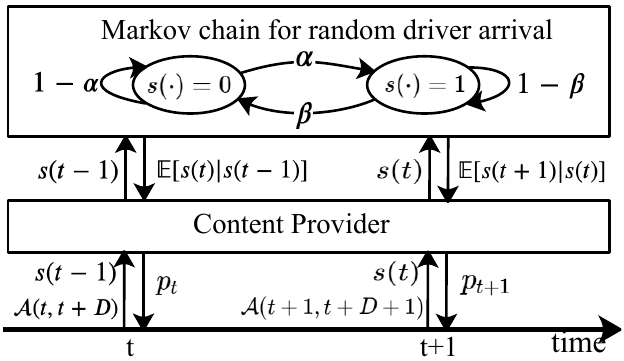}
\caption{Information observation and decision process by the online provider.}
\label{event}
\end{figure}

Besides the random arrival process, we consider the challenging incomplete information scenario that the content provider does not know each driver's actual cost when deciding pricing $p_t(A(t+D))$ before observing $s(t)$. A driver has extra cost to travel on the non-shortest path with delay $D$, and different drivers have different cost sensitivities. Let $x\in[0,1]$ be the normalized cost sensitivity of a driver, then we model its cost as $\text{cost}=xD$, which is proportional to the delay with individual cost sensitivity $x$. Thus, a driver with greater $x$ is more cost sensitive and less likely to accept the price offer to change route. Yet, the provider only knows that $x$ in a normalized range $[0,1]$ randomly follows cumulative distribution function (CDF) $F(\cdot)$, which can be obtained by fitting different historical data's frequencies into a histogram and converting it to CDF as \cite{meijer2006measuring}. The maximum travel cost is thus $D$ and the content provider will not decide any pricing reward $p_t$ beyond $D$ to over-pay any driver, and we expect
\begin{equation*}
    0\leq p_t\big(A(t+D)\big)\leq D.
\end{equation*}

The cost sensitivity distribution in certain applications is considered to be truncated normal or logistic random distribution \cite{meijer2006measuring,kim2010modeling}, which are used for our simulations later. In this paper, we assume that the i.i.d. random distributions $F(\cdot)$ of drivers' cost sensitivities satisfies the following assumption, which is the regular value distribution widely assumed in the mechanism design literature \cite{ewerhart2013regular}. 

\begin{assumption}\label{assumption1}
Assume that the following function
\begin{equation}
    H(x):=x+\frac{F(x)}{F'(x)}\label{H(x)}
\end{equation}
increases monotonically in $x\in[0,1]$, where $F(x)$ and $F'(x)$ are the CDF and PDF of a driver's random cost sensitivity. 
\end{assumption}

According to \cite{ewerhart2013regular}, the second term $\frac{F(x)}{F'(x)}$ of (\ref{H(x)}) is the complementary hazard rate and the monotonicity of $H(x)$ tells the Myerson's regularity. Assumption \ref{assumption1} is not strong and it is satisfied by multiple common distributions such as uniform, exponential and logistic distributions. Later in Section \ref{section3} we will relax Assumption \ref{assumption1} and extend our analysis and results to some other distributions (e.g., truncated normal distribution) in Corollary \ref{corol:distribution}.

\subsection{Online Pricing Problem Formulation}  \label{section3b}
Upon arrival at time $t$ with $s(t)=1$, a driver observes the current pricing reward $p_t(A(t+D))$. He decides to accept this offer or not, by checking if his utility
\begin{equation}
    U_t\big(A(t+D)\big)=p_t\big(A(t+D)\big)-xD \label{Ut}
\end{equation}
of travelling on the non-shortest path is positive or not. If $s(t)=1$ and $U_t\big(p_t(A(t+D))\big)\geq 0$ at time $t$, the driver accepts the offer at the gateway X, and updates the whole path's information after travel delay $D$, helping reduce the AoI $A(t+D)$ of this path to $D$, as shown in the right part of Fig. \ref{Single_Path}. Otherwise, the next foreseen AoI $A(t+D+1)$ at time $t+D+1$ increases from $A(t+D)$ by one slot to $A(t+D)+1$. Thus, the dynamics of the actual AoI is given as
\begin{equation}
    A(t+D+1)=
    \begin{cases}
    D, &\text{if}\ U_t\big(A(t+D)\big)\geq 0\\& \text{and}\ s(t)=1;\\
    A(t+D)+1,& \text{otherwise,}\label{A(t+D+1)}
    \end{cases}
\end{equation}
where $U_t\big(A(t+D)\big)$ is defined in (\ref{Ut}), and the pricing reward is accepted in the former case with probability 
\begin{equation}
    Q(t)=\mathbb{E}[s(t)|s(t-1)]F\left(\frac{p_t(A(t+D))}{D}\right) \label{QD_t}
\end{equation}
with $\mathbb{E}[s(t)|s(t-1)]$ in (\ref{Est}) from the content provider's point of view. The expectation of the final payment to the possibly arrived driver at time $t$ is thus $Q(t)p_t(A(t+D))$, and the online pricing $p_t(A(t+D))$ needs to well balance the next AoI evolution in (\ref{A(t+D+1)}) and the expected payment. 

Finally, we are ready to formulate the provider's online pricing problem in Fig. \ref{event} by the Markov decision process (MDP) techniques with four components \cite{puterman2014markov,hsu2017age} to best adapt to the actual AoI evolution $A(t+D)$ and past arrival observation $s(t-1)$ under the memoryless property. 
\begin{itemize}
    \item \textbf{States:} We define the state of the MDP in slot $t$ by the tuple $\mathbf{S}_t=(A(t+D),s(t-1))$. Note that $s(-1)=0$ at initial time slot $t=0$.
    \item \textbf{Actions:} The action of the MDP in slot $t$ is the price $p_t(A(t+D))\in[0,D]$. Note that the continuous action space size is infinite.
    \item \textbf{Transition probabilities:} According to the conditional arrival probability in (\ref{Est}) and AoI dynamics in (\ref{A(t+D+1)}), there can be three possible states at the next time slot $t+1$. The path will be sampled by an arrival driver with probability $Q(t)$ in (\ref{QD_t}). If there is no arrival at time $t$, the state $\mathbf{S}_{t+1}$ at $t+1$ will be $(A(t+D)+1,0)$ with probability
    \begin{equation}
        Q_0(t)=1-\mathbb{E}[s(t)|s(t-1)]. \label{Q0_t}
    \end{equation}
    If there is an arrival at time $t$ but the driver does not accept the price, the state $\mathbf{S}_{t+1}$ at $t+1$ will be $(A(t+D)+1,1)$ instead with probability
    \begin{equation}
        Q_1(t)=\mathbb{E}[s(t)|s(t-1)]\bigg(1-F\Big(\frac{p_t(A(t+D))}{D}\Big)\bigg). \label{Q1_t}
    \end{equation} 
    In summary, we can obtain the following state transitions of $\mathbf{S}_{t+1}$:
    \begin{equation}
    \!\!\!\!\!\!\!
        \begin{cases}
            \bar{\mathbf{S}}_{t+1}=(D,1),&\text{with probability } Q(t) \text{ in } (\ref{QD_t}),\\
            \mathbf{S}_{t+1}^{0}=(A(t+D)+1,0),&\text{with probability } Q_0(t) \text{ in } (\ref{Q0_t}),\\
            \mathbf{S}_{t+1}^{1}=(A(t+D)+1,1),&\text{with probability } Q_1(t)\text{ in } (\ref{Q1_t}),
        \end{cases}\label{transition_prob}
    \end{equation}
    which correspond to the three outcomes: arrival to sample, no current arrival, and arrival to not sample.
    \item \textbf{Cost:} Let $V\big(\mathbf{S}_t,p_t(A(t+D))\big)$ be the immediate cost of the MDP if action $p_t(A(t+D))$ is taken in slot $t$ under state $\mathbf{S}_t$, which is defined as the summation of the actual AoI and expected economic payment:
    \begin{equation}
    \!\!\!\!\!\!
        V\big(\mathbf{S}_t,p_t(A(t+D))\big) =A(t+D)+Q(t) p_t(A(t+D)).\label{cost_V}
    \end{equation}
\end{itemize}
Considering the discount factor $\rho\in(0,1)$ under discrete time horizon, the objective of the MDP is to find the optimal pricing function $p_t(A(t+D))$ at current time slot $t$ that minimizes the long-term total $\rho$-discounted cost:
\begin{equation}
    \begin{aligned}
        &C_t^*\big(\mathbf{S}_t\big) =\min_{p_t(A(t+D))\in[0,D]} \sum_{\tau=t}^{T} \rho^{\tau-t} V\Big(\mathbf{S}_t,p_t\big(A(t+D)\big)\Big),\\
        &\quad \quad\quad \quad\quad \quad \ s.t.\quad \quad \ \  (\ref{Est}),(\ref{Ut}),(\ref{A(t+D+1)})\ \text{and}\ (\ref{cost_V}),
    \end{aligned}\label{MDP_formulation}
\end{equation}
which is a non-convex problem due to the non-convex AoI dynamics constraint. The current price $p_t(A(t+D))$ to announce affects the dynamics of AoI since $t+1$. Note that the current pricing decision can only affect the AoI after a time delay $D$, thus like $p_t(A(t+D))$ we also use the foreseen AoI $A(t+D)$ in the cost objective function $C_t^*\big(\mathbf{S}_t\big)$ in (\ref{MDP_formulation}) above.

\subsection{Characterization of Dynamic Program}

For any time $t=\{0,\cdots,T-D\}$ in the finite time horizon, problem (\ref{MDP_formulation}) can be written as \cite{puterman2014markov}:
\begin{equation}
\begin{aligned}
   &\Scale[0.94]{C_t^*\big(\mathbf{S}_t\big)}=\min_{p_t(A(t+D))}\Scale[0.94]{V\Big(\mathbf{S}_t,p_t\big(A(t+D)\big)\Big) +\rho\mathbb{E}_{\mathbf{S}_{t+1}}\big[C^*_{t+1}\big(\mathbf{S}_{t+1}\big)\big]},
\end{aligned}\label{HJB}
\end{equation}
where the cost-to-go is
\begin{align*}
    \mathbb{E}_{\mathbf{S}_{t+1}}\big[C^*_{t+1}\big(\mathbf{S}_{t+1}\big)\big]=\sum_{j=0}^1Q_j(t) C_{t+1}^*\big(\mathbf{S}_{t+1}^{j}\big)+Q(t) C_{t+1}^*\big(\bar{\mathbf{S}}_{t+1}\big)
\end{align*}
according to the transition probabilities (\ref{transition_prob}).

The optimal pricing satisfies the first-order necessary condition of (\ref{HJB}) with respect to $p_t(A(t+D))$, i.e.,
\begin{equation}
\begin{aligned}
     &\frac{p_t^*(A(t+D))}{D}+\frac{F\big(\frac{p_t^*(A(t+D))}{D}\big)}{F'\big(\frac{p_t^*(A(t+D))}{D}\big)}- \frac{\rho}{D} \Delta C_{t+1}^*=0,\label{u*t}
\end{aligned}
\end{equation}
where 
\begin{equation}
    \Delta C_{t+1}^*=C_{t+1}^*\big(\mathbf{S}_{t+1}^{1}\big)-C_{t+1}^*\Big(\bar{\mathbf{S}}_{t+1}\Big).\label{delta_C}
\end{equation}
To solve (\ref{u*t}), we need the inputs of  $C_{t+1}^*\big(\mathbf{S}_{t+1}^{1}\big)$ and $C_{t+1}^*\big(\bar{\mathbf{S}}_{t+1}\big)$. For each long-term cost function since time $t$, it takes $O\big(\log_2\big(\frac{D}{\varepsilon}\big)\big)$ complexity to obtain its pricing solution to (\ref{u*t}) using binary search with error $\varepsilon$ \cite{knuth1971optimum}. 

Lacking the information about drivers' hidden sampling costs and even arrivals, the transition probabilities are related with the pricing $p_t^*(A(t+D))$ in (\ref{HJB}). Therefore, it is infeasible to apply online control methods such as Hamilton-Jacobi-Bellman equations and neural networks approximation here \cite{al2008discrete,dierks2010optimal,lewis2012optimal}.

\begin{table}[!t]
\renewcommand{\arraystretch}{1.3}
\caption{Key notations used in the paper}
\label{notation_table}
\centering
\begin{tabular}{c|m{0.33\textwidth}}
\hline
\textbf{Notation} & \textbf{Definition}\\
\hline
\hline
$D$ & The extra travel delay of the single non-shortest path\\
\hline
$T$ & The time horizon\\
\hline
$A(t+D)$ & The foreseen AoI with the travel delay $D$ at the beginning of time slot $t$\\
\hline
$p_t(\cdot)$ & The pricing reward to compensate an arrived driver at time $t$\\
\hline
$s(t-1)$ & Driver's arrival information of the last time slot $t-1$\\
\hline
$\mathbb{E}[s(t)|s(t-1)]$ & The conditional probability of a driver's arrival during time slot $t$.\\
\hline
$x$ & The cost sensitivity of a driver\\
\hline
$F(\cdot)$ & The CDF of a driver's random cost sensitivity\\
\hline
$U_t(\cdot)$ & The utility of a driver at time slot $t$\\
\hline
$\mathbf{S}_t$ & The state $(A(t+D),s(t-1))$ at time $t$ in the single path problem\\
\hline
Q(t) & Transition probabilities at time slot $t$ in MDP\\
\hline
$V(\cdot)$ & The immediate cost of the MDP\\
\hline
$\rho$ & The discount factor\\
\hline
$C_t(\cdot)$ & The $\rho$-discounted long-term cost function at time $t$\\
\hline
$\varepsilon$ & The error of binary search\\
\hline
$\mathbf{C}^*(\bar{\mathbf{S}})$ & The look-up table to store cost function $C_t^*(\bar{\mathbf{S}}_t)$ for any $t$ \\
\hline
$N$ & The total number of non-shortest paths\\
\hline
$D_i$ & The travel delay of the $i$-th non-shortest path\\
\hline
$\hat{i}_{t}$ & The index of the path with the maximum net payoff at time slot $t$\\
\hline
$\mathbb{A}(t)$ & The foreseen AoI set of all the non-shortest paths at the beginning of time slot $t$\\
\hline
$\mathbb{S}_t$ & The state $(\mathbb{A}(t),\hat{i}_{t-1},s(t-1))$ at time $t$ in the multi-path problem\\
\hline
\end{tabular}
\end{table}

As it is too late to price after $t=T-D$ for returning timely information before end-time $T$, we obtain from (\ref{HJB}) that 
\begin{equation}
    C_{T-D}^*\big(\mathbf{S}_{T-D}\big)=A(T).\label{VT}
\end{equation}
Based on the state transitions in (\ref{transition_prob}) and equation (\ref{VT}), we find that the state space increases polynomially with time horizon $T$ in (\ref{HJB}) for the single-path pricing problem. Thus, one can use backward induction to solve the problem (\ref{HJB}) in polynomial time \cite{puterman2014markov,zhou2019collaborative}, which may be not small for a large $T$. Thus, we will further exploit the unique AoI feature and propose Algorithms \ref{lookup_table} and \ref{backward_computation} to reduce to linear complexity later in Section \ref{section4}. On the other hand, we will show later in Section~\ref{section5} that the state space increases exponentially in $T$ and path number $N$, which forbids to apply traditional MDP methods \cite{bellman1966dynamic}.

Before that, we first prove key structural properties of the optimal (long-term) cost functions in (\ref{HJB}) and optimal online pricing in next section. For ease of reading, we summarize all the key notations in Table \ref{notation_table}.

\section{Proved Structural Properties of Optimal Online Pricing}\label{section3}
Since the price $p_t\big(A(t+D)\big)$ is located in a compact interval $[0,D]$, there always exists optimal pricing solutions to the problem (\ref{HJB}). However, since drivers' information about costs and random arrivals are incomplete, the monotonicity of cost function with respect to $A(t+D)$ and the uniqueness of optimal pricing solution to (\ref{u*t}) cannot be guaranteed \cite{kakumanu1971continuously} and \cite{puterman2014markov}. In this section, at any time $t$, we first examine the monotonic relationship of cost function and optimal pricing with respect to the foreseen AoI $A(t+D)$. Then we prove the uniqueness of optimal online solution to problem (\ref{HJB}). 

By observing (\ref{u*t}), we find the optimal pricing solution $p^*_t\big(A(t+D)\big)$ is determined by the the cost functions $C_{t+1}^*\big(\mathbf{S}_{t+1}^{1}\big)$ without information update, $C_{t+1}^*(\bar{\mathbf{S}}_{t+1})$ with update, and the CDF $F(\cdot)$. To examine the optimal pricing's relationship with $A(t+D)$, we first need to examine the relationship between cost function $C^*_t\big(\mathbf{S}_{t}\big)$ in (\ref{HJB}) and foreseen AoI $A(t+D)$ \cite{puterman2014markov}.

\begin{proposition}\label{C_increas}
The optimal cost function $C^*_t\big(\mathbf{S}_t\big)$ in (\ref{HJB}) at any time $t\in\{0,1,\cdots,T-D\}$ increases monotonically with the foreseen AoI $A(t+D)$.
\end{proposition}
\begin{proof}
Suppose that there are two states $\mathbf{S}_t^a$ and $\mathbf{S}_t^b$ with their actual AoI satisfying $A^a(t+D)> A^b(t+D)$ at any time slot $t\in\{0,\cdots, T-D\}$. In the following, we apply both mathematical induction and backward induction to prove
\begin{equation}
    C^*_{t}\big(\mathbf{S}_{t}^a\big)>C^*_{t}\big(\mathbf{S}_{t}^b\big), \label{maths_induction}
\end{equation}
which can tell the monotonicity of the optimal cost function.

We first prove the base case at the last time slot $t=T-D$. According to (\ref{VT}), the two cost functions satisfy
\[
C_{T-D}^*\big(\mathbf{S}_{T-D}^{a}\big)=A^a(T)>A^b(T)=C_{T-D}^*\big(\mathbf{S}_{T-D}^{b}\big),
\]
which holds for (\ref{maths_induction}).

Next, we assume the induction hypothesis that $C_{t_0+1}^*\big(\mathbf{S}_{t_0+1}^{a}\big)>C_{t_0+1}^*\big(\mathbf{S}_{t_0+1}^{b}\big)$ is true for a particular time slot $t_0+1$. It follows to show that $C_{t_0}^*\big(\mathbf{S}_{t_0}^{a}\big)>C_{t_0}^*\big(\mathbf{S}_{t_0}^{b}\big)$ also holds.
Denote $p_{t_0}^*(A^a(t_0+D))$ and $p_{t_0}^*(A^b(t_0+D))$ to be the optimal pricing to $A^a(t_0+D)$ and $A^b(t_0+D)$, respectively. Consider another non-optimal pricing to $A^b(t_0+D)$:
\begin{equation}
    \hat{p}_{t_0}(A^b(t_0+D))=p_{t_0}^*(A^a(t_0+D)),\label{hat_p}
\end{equation}
whose corresponding non-optimal cost function must satisfy
$\hat{C}_{t_0}\big(\mathbf{S}_{t_0}^{b}\big)\geq C_{t_0}^*\big(\mathbf{S}_{t_0}^{b}\big)$.
Then we can derive the induction step: 
\begin{equation}\notag
    \begin{aligned}
        &C_{t_0}^*\big(\mathbf{S}_{t_0}^{a}\big)-C_{t_0}^*\big(\mathbf{S}_{t_0}^{b}\big)\geq C_{t_0}^*\big(\mathbf{S}_{t_0}^{a}\big)-\hat{C}_{t_0}\big(\mathbf{S}_{t_0}^{b}\big)\\
        = &A^a(t+D)-A^b(t+D)\\&+\rho \mathbb{E}_{\mathbf{S}_{t_0+1}^a}\left[C_{t_0+1}^*(\mathbf{S}_{t_0+1}^a)\right]-\rho\mathbb{E}_{\mathbf{S}_{t_0+1}^b}\left[C_{t_0+1}^*(\mathbf{S}_{t_0+1}^b)\right],
    \end{aligned}
\end{equation}
which is larger than $0$ due to the facts that $A^a(t+D)>A^b(t+D)$, $C^*_{t_0+1}\big(\mathbf{S}_{t_0+1}^a\big)\geq C^*_{t_0+1}\big(\mathbf{S}_{t_0+1}^b\big)$ and (\ref{hat_p}).

As we have proven that both the base case and the induction step are true, (\ref{maths_induction}) holds for each $t\in\{0,\cdots,T-D\}$ by mathematical induction. This further tells the monotonicity of the optimal cost function $C^*_t\big(\mathbf{S}_t\big)$ in (\ref{HJB}).
\end{proof}

Intuitively, a larger foreseen AoI $A(t+D)$ leads to larger long-term cost. 
Based on the monitonicity of the long-term cost function in Proposition \ref{C_increas}, next we are ready to prove the uniqueness of optimal online pricing solution to problem (\ref{HJB}) at any time $t$.
\begin{proposition}\label{uniquesol_prop}
Under Assumption \ref{assumption1}, the optimal online pricing solution $p_t^*(A(t+D))$ to problem (\ref{HJB}) exists and is unique at any time $t\in\{0,\cdots,T-D\}$.
\end{proposition}
\begin{proof}
Let $y=\frac{p_t(A(t+D))}{D}\in[0,1]$ and define $G_t(y)$ to be
\begin{equation}\notag
    G_t(y):=H(y)-\frac{\rho}{D}\Delta C^*_{t+1}
\end{equation}
where $H(\cdot)$ is defined in (\ref{H(x)}) and $\Delta C^*_{t+1}$ is in (\ref{delta_C}). According to (\ref{u*t}), if there exists $y^*\in[0,1]$ such that 
\[G_t(y^*)=0\]
at any time $t$, its corresponding $p^*_t(A(t+D))$ is the optimal pricing to (\ref{HJB}).

To verify the existence of $y^*$, we next check the two endpoints $G_t(0)$ and $G_t(1)$. As $H(0)=0$ and the two cost functions satisfy \[C_{t+1}^*\big(\mathbf{S}^{1}_{t+1}\big)>C_{t+1}^*(\bar{\mathbf{S}}_{t+1})\] 
by Proposition \ref{C_increas}, we can obtain $G_t(0)<0$. 

If $G_t(1)\geq 0$, i.e., $H(1)\geq \frac{\rho}{D}\Delta C^*_{t+1}$, there exists a solution of $y$ in $[0,1]$. 
Under Assumption \ref{assumption1}, $H(y)$ is an increasing function, and thus $G_t(y)$ increase with $y$ on $[0,1]$, which means the solution is unique. 

If $G_t(1)<0$ instead, there is no solution in $[0,1]$, but we can still uniquely set
\[p_t^*(A(t+D))=D\] 
to hurdle the large foreseen AoI. 
\end{proof} 

We can generalize our Proposition \ref{uniquesol_prop} to fit some other distributions such as truncated normal distribution. 
\begin{corollary} \label{corol:distribution}
If there exists a cutoff point $x_{th}\in[0,1)$, such that $H(x)<0$ for $x\in [0, x_{th}]$ and $H(x)$ increases with $x\in[x_{th},1]$, which holds for truncated normal distribution, then the optimal pricing solution to problem (\ref{HJB}) exists and is unique. 
\end{corollary}

Based on Propositions \ref{C_increas}, \ref{uniquesol_prop} and the expression (\ref{u*t}), we can also derive the monotonicity of optimal online pricing $p_t^*(A(t+D))$ with respect to $A(t+D)$ \cite{puterman2014markov}. 
\begin{corollary}\label{Ath_prop}
The unique optimal online pricing $p_t^*(A(t+D))$ at any time $t\in\{0,\cdots,T-D\}$ increases monotonically with the foreseen AoI $A(t+D)$.
\end{corollary}

Corollary \ref{Ath_prop} tells that as the foreseen AoI is large to dominate in the long-term cost objective (\ref{HJB}), we should set a high price to immediately motivate drivers to sample and control the AoI evolution.

\section{Linear-time algorithm for optimal online pricing}\label{section4}

\begin{figure}[t]
    \centering
    \subfigure[Original state transitions from $t=0$ to $t=2$.]{
    \label{inductive_D}
    \includegraphics[width=0.45\textwidth]{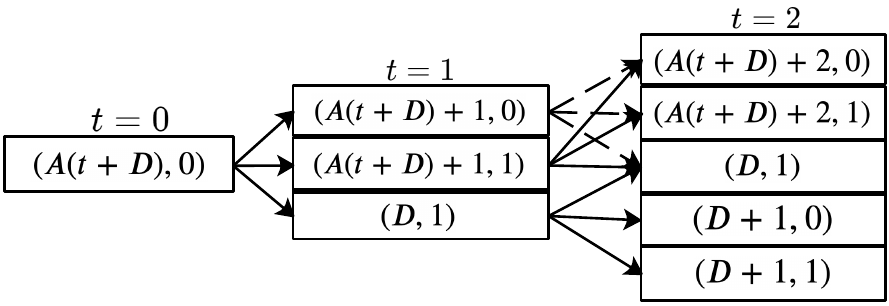}}
    \subfigure[Look-up table to further simplify the state space above.]{
    \label{inductive_compute}
    \includegraphics[width=0.45\textwidth]{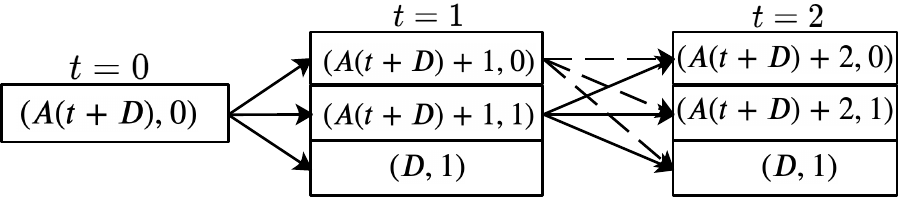}}
    \caption{An illustrative example of $T-D=2$ for the look-up table approach to reduce state space.}
\end{figure}

Recall that backward induction can be used to solve problem (\ref{HJB}) in polynomial time. In this section, we exploit the dynamic feature of AoI to jointly use a fixed look-up table and backward induction to further reduce the increasing number of cost functions in the time domain. Furthermore, we show our designed algorithm is also applicable to infinite time horizon.

According to (\ref{u*t}), at any time $t$ we first need to compute long-term cost functions $C^*_{t+1}(\mathbf{S}_{t+1}^{1})$ and $C^*_{t+1}\big(\bar{\mathbf{S}}_{t+1}\big)$ before solving optimal pricing $p_t^*(A(t+D))$. To obtain $C^*_{t+1}\big(\mathbf{S}_{t+1}^{1}\big)$ for example, as discussed at the end of Section \ref{section2}, we further need to apply iteration computation to derive another three cost functions $C^*_{t+2}\big(\mathbf{S}_{t+2}^{0}\big)$, $C^*_{t+2}\big(\mathbf{S}_{t+2}^{1}\big)$ and $C^*_{t+2}\big(\bar{\mathbf{S}}_{t+2}\big)$ at time $t+2$. After careful examination of this branching network, we find out that a large number of cost functions repeatedly appear due to the linearly increasing feature of AoI over time. Take the transition of $(A(t+D),0)$ at time $t=0$ with $T-D=2$ as an illustrative example, which implies three decision slots $t=0,1,2$. Fig. \ref{inductive_D} shows the polynomial increase of state space for iterative computation from $t=0$ to $t=2$. At time $t=1$, $(A(t+D)+1,0)$ with no driver arrival during last time slot $t=0$ (i.e., $s(0)=0$) and $(A(t+D)+1,1)$ with last arrival have the same branches to $(A(t+D)+2,0),(A(t+D)+2,1)$ and $(D,1)$ at $t=2$. As there are actually $5$ states at time $t=2$, we can apply backward induction as \cite{bellman1966dynamic} and \cite{zhou2019collaborative} to derive the optimal online pricing at current time $t$ in polynomial time.

Moreover, as inspired by \cite{bemporad2002explicit} and \cite{wang2009fast}, we want to further simplify the state space by applying the look-up table approach. As shown in Fig. \ref{inductive_D}, we find the state $\bar{\mathbf{S}}_t=(D,1)$ appears at each time slot $t\in \{0, \cdots, T-D\}$ as it is possible to find the driver arrival to sample and update information at any time. Since the inputs $D$ and $1$ of these functions are constant, we do not need to wait till our observation of such constants to decide the pricing online. Instead, we propose to compute any $C^*_t(\bar{\mathbf{S}}_t)$ before $t=0$ and store them in a fixed look-up table $\mathbf{C}^*(\bar{\mathbf{S}})$ for any online use later. Then we no longer need to expand the state space from $\bar{\mathbf{S}}_t=(D,1)$ at any time $t$. As illustrated in Fig. \ref{inductive_compute}, by jointly applying the look-up table, we end up with only three cost functions to compute online at both $t=1$ and $t=2$. Thus, we successfully reduce the number of cost functions to compute online from quadratic term $(T-D)^2$ to linear item $3(T-D)$ at each time $t$.

\begin{algorithm}[t]
\caption{The computation of look-up table $\mathbf{C}^*(\bar{\mathbf{S}})$ by backward induction}
\label{lookup_table}
\begin{algorithmic}[1]
\STATE Store $C_{T-D}^*(D,1)=D$ in $\mathbf{C}^*(\bar{\mathbf{S}})$table;\ \label{algo1_step1}
\FOR{$i=1$ to $i=T-D$} \label{algo1_step2}
\STATE Compute $C_{T-D}^*(D+i,1)$ and $C_{T-D}^*(D+i,0)$ according to (\ref{VT});\ \label{algo1_step3}
\FOR{$j=1$ to $i$} \label{algo1_step4}
\STATE Obtain $C_{T-D-j+1}^*(D,1)$ from $\mathbf{C}^*(\bar{\mathbf{S}})$ table;\
\STATE Compute $p_{T-D-j}^*(D+i-j)$ according to (\ref{u*t}) using binary search with error $\varepsilon$;\
\STATE Compute $C_{T-D-j}^*(D+i-j,1)$ and $C_{T-D-j}^*(D+i-j,0)$ according to (\ref{HJB}) from last decision slot;\ 
\ENDFOR \label{algo1_step8}
\STATE Store $C_{T-D-i}^*(D,1)$ in $\mathbf{C}^*(\bar{\mathbf{S}})$ table;\ \label{algo1_step9}
\ENDFOR \label{algo1_step10}
\RETURN $\mathbf{C}^*(\bar{\mathbf{S}})$ table
\end{algorithmic}
\end{algorithm}

\begin{algorithm}[t]
\caption{The linear-time computation of online optimal pricing $p_t^*(A(t+D))$ to (\ref{HJB}) at any time $t\in\{0,\cdots,T-D\}$}
\label{backward_computation}
\begin{algorithmic}[1]
\STATE \textbf{Input: $\mathbf{C}^*(\bar{\mathbf{S}})$ table, binary search error $\varepsilon$};\ \label{input_table}
\STATE Observe $A(t)$ and predict foreseen AoI $A(t+D)$;\ 
\STATE Compute $C_{T-D}^*\big(\mathbf{S}_{T-D}^{0}\big)$ and $C_{T-D}^*\big(\mathbf{S}_{T-D}^{1}\big)$ according to (\ref{VT});\ \label{algo2_step3}
\FOR{$1\leq i\leq T-D-t-1$} \label{algo2_step4}
\STATE Obtain $C_{T-D-i}^*(\bar{\mathbf{S}}_{T-D-i})$ from $\mathbf{C}^*(\bar{\mathbf{S}})$ table;\ \label{step_table}
\STATE Compute $C_{T-D-i}^*\big(\mathbf{S}_{T-D-i}^{0}\big),C_{T-D-i}^*\big(\mathbf{S}_{T-D-i}^{1}\big)$ according to (\ref{HJB}) from last decision slot $T-D-i+1$;\ \label{cluster_alg2}
\STATE Compute $p_{T-D-1-i}^*\big(A(t+D)+T-D-1-i\big)$ according to (\ref{u*t}) using binary search with error $\varepsilon$;\
\label{search_error}
\ENDFOR \label{algo2_step9}
\RETURN $p^*_t(A(t+D))$ for current time $t$ \label{online_2}
\end{algorithmic}
\end{algorithm}

Thanks to the look-up table approach that successfully reduces the polynomially increasing number of cost functions to linear number, we are ready to propose Algorithm \ref{lookup_table} to compute the look-up table $\mathbf{C}^*(\bar{\mathbf{S}})$ and Algorithm \ref{backward_computation} to optimally solve problem (\ref{HJB}) for online pricing at any time $t=\{0,\cdots,T-D\}$. 

In Algorithm \ref{lookup_table}, we apply backward induction to calculate cost functions $C^*_t(D,1)$ from $t=T-D$ to $t=0$:
\begin{itemize}
    \item In step \ref{algo1_step1}, we first obtain the value of $C_{T-D}^*(D,1)$ in time slot $T-D$ according to (\ref{VT}).
    \item From step \ref{algo1_step2} to \ref{algo1_step10}, we use the loop to calculate $C_t^*(D,1)$ backward from time slot $t=T-D-1$ to $t=0$, and store each value in the $\mathbf{C}^*(\bar{\mathbf{S}})$ table in step \ref{algo1_step9}.
    \item From step \ref{algo1_step4} to \ref{algo1_step8}, to calculate each $C^*_{T-D-i}(D,1)$ in the $i$-th loop, we calculate all the possible long-term cost functions backward from time slot $t=T-D$ to $t=T-D-i+1$, according to Fig. \ref{inductive_D}.
\end{itemize}
Note that here we use the state tuple $(D,1)$ in $C^*_t(\cdot)$ to distinguish from the state space $\bar{\mathbf{S}}_t$ elsewhere. After returning the fixed look-up table $\mathbf{C}^*(\bar{\mathbf{S}})$, we are ready to apply it to solve $p^*(A(t+D))$ in Algorithm \ref{backward_computation}:
\begin{itemize}
    \item In step \ref{algo2_step3}, we first obtain the known value of $C_{T-D}^*(\mathbf{S}^{0}_{T-D})$ and $C_{T-D}^*(\mathbf{S}^{1}_{T-D})$ in time slot $T-D$ according to (\ref{VT}), where $\mathbf{S}^{0}_{T-D}$ and $\mathbf{S}^{1}_{T-D}$ are $(A(D)+T-D,0)$ and $(A(D)+T-D,1)$, respectively.
    \item From step \ref{algo2_step4} to \ref{algo2_step9}, we calculate all the possible long-term cost functions and the corresponding optimal price backward from time slot $T-D-t-1$ to $t+1$.
    \item We finally obtain the optimal price $p^*_t(A(t+D))$ during the last loop $i=T-D-t-1$. 
\end{itemize}
As the pricing action space $[0,D]$ is continuous, we equally discretize it into $(\frac{D}{\varepsilon}+1)$ many points with $\varepsilon$ as the gap. Here $\varepsilon$ can also be viewed as the error of binary search to find the optimal online pricing $p_t^*$ according to (\ref{u*t}).

Based on the above analysis, we propose the following theorem and further analyze the complexity of Algorithm \ref{backward_computation}.
\begin{theorem}\label{thm}
Algorithm \ref{backward_computation} returns the optimal online pricing $p_t^*(A(t+D))$ with the linear complexity $O\big((T-D)\log_2\big(\frac{D}{\varepsilon}\big)\big)$ at any time $t=\{0,...,T-D\}$, where $\varepsilon$ is the binary search error in step \ref{search_error} of Algorithm \ref{backward_computation}.
\end{theorem}
\begin{proof}

\begin{figure}[t]
    \centering
    \subfigure[Evolution of actual AoI under the optimal online pricing $p_t^*(A(t+D)) , t\in \{0,\cdots,T=30\}$.]{
    \includegraphics[width=0.4\textwidth]{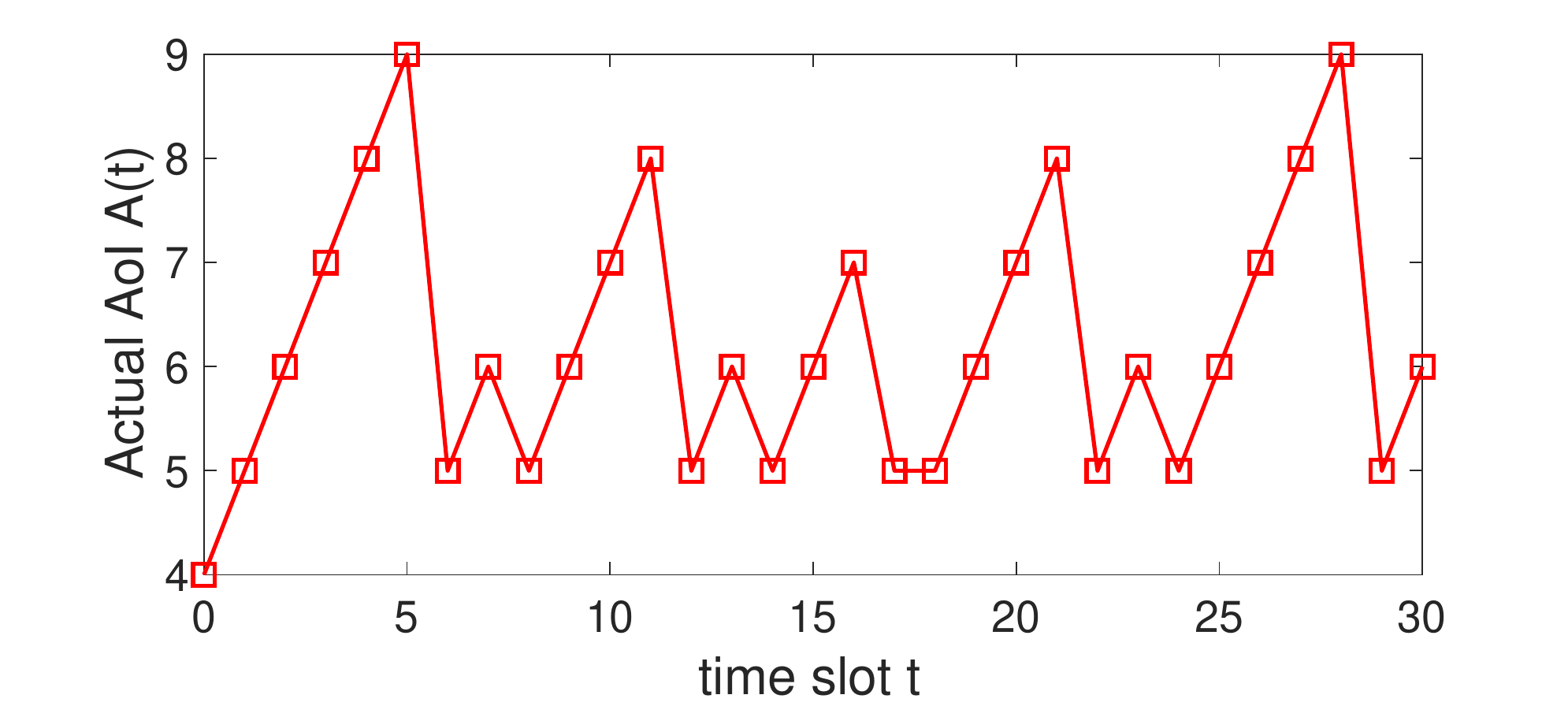}}
    \subfigure[Online optimal pricing policy returned by Algorithm \ref{backward_computation}.]{
    \includegraphics[width=0.4\textwidth]{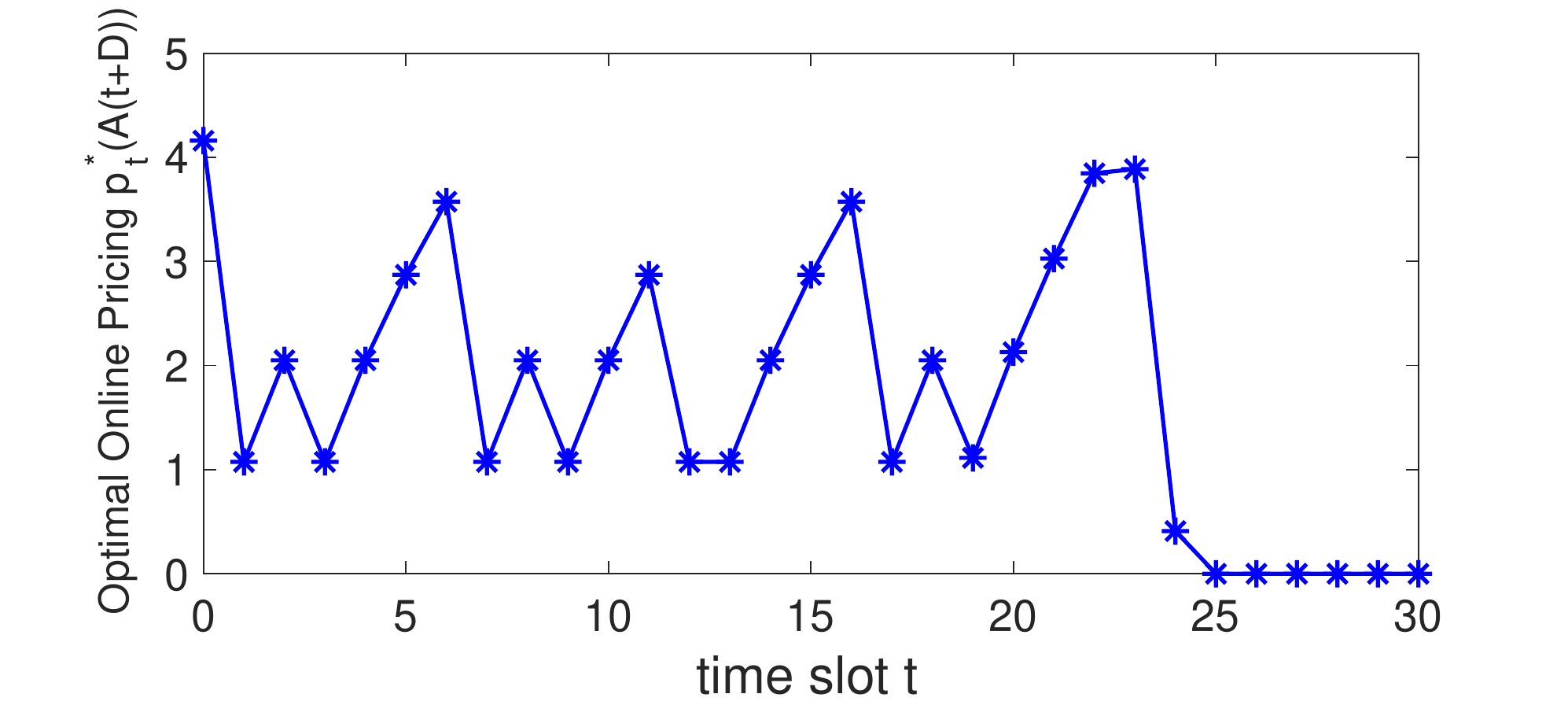}}
    \caption{The evolution of actual AoI $A(t)$ and online optimal price $p^*_t(A(t+D))$ returned by Algorithm \ref{backward_computation} from $t=0$ to $30$. Here we set $A(0)=4, T=30, D=5,\rho=0.85$, probabilities $\alpha=0.8,\beta=0.6$ for the Markov chain in Fig. \ref{event}, and each driver's cost CDF $F(x)$ satisfies a truncated normal distribution with mean $0.6$ and variance $0.7$.}
    \label{pricing_fig}
\end{figure}

As mentioned above, the joint use of backward induction and look-up table reduces the number of cost function to $3(T-D)$ for each time slot, and it takes $\log_2\big(\frac{D}{\varepsilon}\big)$ complexity to obtain each price solution to the cost function by using binary search of (\ref{u*t}). Thus, it takes $O\big((T-D)\log_2\big(\frac{D}{\varepsilon}\big)\big)$ complexity to obtain the optimal pricing for any particular time $t$. 

Note that as the look-up table is computed offline before $t=0$, there is no need to take its complexity into consideration. Then we can say that Algorithm \ref{backward_computation} costs linear-time to calculate the online optimal pricing at any time $t$. 
\end{proof}

Note that storing any other cost functions into a look-up table cannot help to reduce the complexity anymore, because the initial AoI at time $t=0$ is unknown and can be unbounded. Thus, the linear time complexity of Algorithm \ref{backward_computation} is the minimum and cannot be further improved \cite{puterman2014markov,littman2013complexity}.

In the following experiment, we create a typical instance to run the online optimal pricing $p^*_t(A(t+D))$ returned by Algorithm \ref{backward_computation} and corresponding AoI $A(t)$, $\forall t\in\{0, \cdots, T=30\}$ in Fig. \ref{pricing_fig}. Here we set the travel delay $D=5$ for the non-shortest path in Fig. \ref{Single_Path} with initial AoI $A(0)=4$, and the cost distribution of each driver to follow truncated normal distribution with mean $0.6$ and variance $0.7$. Fig. \ref{pricing_fig} shows that the online pricing to announce follows a delayed pattern of the actual AoI observation over time. This is consistent with the monotonicity of the pricing with respect to delayed AoI by $D$ in Corollary \ref{Ath_prop}. We also notice the online pricing converges to zero since $t=T-D$, as any outdated update after $T$ cannot help.

\begin{remark}
Actually, we can revise our algorithm to fit the infinite time horizon $T\rightarrow\infty$, by replacing the $T$ time window by a future finite $T_{th}$ window in Algorithm \ref{lookup_table} and \ref{backward_computation} for efficient computation (with constant complexity) at any time $t$. We can first prove that the optimal cost functions are bounded under $T\rightarrow \infty$, and then prove that the performance error of this approximation algorithm (revised from Algorithm \ref{backward_computation}) exponentially reduces to zero as $T_{th}$ increases.
\end{remark}

\section{Generalization of online pricing to Multi-Path Sampling Scenario}\label{section5}
For ease of exposition, we only consider to sample one non-shortest path in Fig. \ref{Single_Path} and propose a linear algorithm to solve the problem. In this section, we generalize to sample an arbitrary number $N$ of different non-shortest paths and introduce the generalized system model first. Different from the single-path problem, the state space here increases exponentially with time horizon $T$ and path number $N$. To overcome the curse of dimensionality in the spatial path domain (i.e., the path number), we first prove that it is optimal to only price one path at a time, which is not necessarily the path with the largest current AoI. We then propose a new backward-clustered computation method across paths and design an approximation algorithm of polynomial complexity to alternate different paths to price over time.

\subsection{Model extension to sample multiple diverse paths}

Upon arrival at the gateway X by following the Markov chain, a driver now faces an arbitrary number $N$ of non-shortest paths from X to destination point Y for routing. We consider diverse paths such that each path $i\in\{1,\cdots,N\}$ has a unique travel delay $D_i$ and a driver with personalized cost sensitivity $x\sim F(\cdot)$ incurs
\begin{equation*}
    \text{cost}=xD_i 
\end{equation*}
to travel on this path to sample fresh information there. We aim to design online pricing to all the paths to help sample fresh information globally to control the maximum AoI. We summarize the actual AoI evolutions in all the $N$ paths in set
\begin{equation}
    \mathbb{A}(t)=\big\{A_i(t+D_i)|i\in \{1,...,N\}\big\}\label{At_multi}
\end{equation}
to tell the maximum foreseen AoI from time $t$, depending on which path $\hat{i}_{t-1}$ was sent by the last driver (if any) at time $t-1$ to sample. We next see how does $\mathbb{A}(t)$ update for all the paths.

At the beginning of time $t\in \{0, \cdots, T\}$, the provider decides online pricing set
\begin{equation}
    \mathbb{P}\big(\mathbb{A}(t)\big)=\big\{p_{i,t}\big(\mathbb{A}(t)\big)\in[0,D_i]|i\in\{1,2,...,N\}\big\} \label{P_set}
\end{equation}
based on the foreseen AoI set $\mathbb{A}(t)$. Given the online prices to different paths, a driver with cost sensitivity $x$ (if arrives at time $t$) finds path $i$ appealing as long as the price there can justify his travel cost, i.e., his utility of travelling path $i$
\begin{equation}
    U_{i,t}(\mathbb{A}(t))=p_{i,t}\big(\mathbb{A}(t)\big)- x D_i \label{Ut_N}
\end{equation} 
is no less than 0, which depends on the pricing set $\mathbb{P}\big(\mathbb{A}(t)\big)$ and the driver's private travel cost sensitivity $x$. If he finds multiple paths appealing, he will optimally accept the path $\hat{i}_t$ with maximum net payoff, i.e.,
\begin{equation}
    \hat{i}_{t}=\begin{cases}
    \arg\max_{i\in\{1,\cdots,N\}} U_{i,t}(\mathbb{A}(t)), &\text{if }\max_i U_{i,t}(\mathbb{A}(t))\geq 0,\\
    0,& \text{otherwise}.\label{hat_i}
    \end{cases}
\end{equation}
If there is no driver arrival at time $t$ or this driver does not consider to sample any of the $N$ non-shortest paths, we set $\hat{i}_{t}=0$. Based on (\ref{hat_i}), we update all the paths' next foreseen AoI as:  
\begin{equation}
    A_i(t+D_i+1)=\left\{
    \begin{aligned}
        & D_i, &\text{if }i=\hat{i}_{t},\\
        & A_i(t+D_i)+1, &\text{if }i\neq \hat{i}_{t}.\label{A(i)}
    \end{aligned}
    \right .
\end{equation}

Based on the above model extension, we are ready to apply the MDP formulation as \cite{puterman2014markov} and \cite{hsu2017age} again to formulate the provider's online pricing problem below.
\begin{itemize}
    \item \textbf{States:} We define the state of the new MDP in slot $t$ by $\mathbb{S}_t=(\mathbb{A}(t),\hat{i}_{t-1},s(t-1))$. Similarly, $s(-1)=0$ and $\hat{i}_{-1}=0$ at time $t=0$.
    \item \textbf{Actions:} The action of the MDP in slot $t$ is the price set $\mathbb{P}\big(\mathbb{A}(t)\big)$ defined in (\ref{P_set}).
    \item \textbf{Transition probabilities:} Compared to the state transitions (\ref{transition_prob}) in single-path, the state space of $\mathbb{S}_t$ now increases exponentially with $N$, because every path has probability to be sampled at each time slot. Denote the probability that path $i$ is accepted by an arrived driver to sample (i.e., $i=\hat{i}$) by $q_i(t)$. The state $\mathbb{S}_{t+1}=(\mathbb{A}(t+1),\hat{i}_{t},s(t))$ at time slot $t+1$ can change into
    \begin{equation}
    \!\!\!\!\!\!\!\!\!
        \begin{cases}
            \mathbb{S}_{t+1}^{a}=(\mathbb{A}(t)+1,0,0),&\text{with probability } \\ & \mathbf{Q}_a(t)=1-\mathbb{E}[s(t)|s(t-1)],\\
            \mathbb{S}_{t+1}^{b}=(\mathbb{A}(t)+1,0,1),&\text{with probability }\mathbf{Q}_b(t)=\\ & \mathbb{E}[s(t)|s(t-1)]\big(1-\sum_{i=1}^Nq_i(t)\big),\\
            \mathbb{S}_{t+1}^{i}=(\mathbb{A}(t+1),i,1),&\text{with probability}\\ & \mathbf{Q}_i(t)=\mathbb{E}[s(t)|s(t-1)]q_i(t),
        \end{cases}\label{transition_prob_M}
    \end{equation}
    where the dynamics of $\mathbb{A}(t+1)$ is defined in (\ref{A(i)}). There are totally $N+2$ outcomes: no current arrival, arrival to not sample, and arrival to sample path $i\in\{1,\cdots,N\}$.
    \item \textbf{Cost:} Let $V\big(\mathbb{S}_t,\mathbb{P}\big(\mathbb{A}(t)\big)\big)$ be the immediate cost of the MDP if action $\mathbb{P}\big(\mathbb{A}(t)\big)$ is taken in slot $t$ under state $\mathbb{S}_t$, which concerns about the maximum foreseen AoI among all the paths as \cite{kadota2018scheduling} and \cite{he2016optimal}:
    \begin{equation}
    \!\!\!\!\!\!\!\!\!
        V\big(\mathbb{S}_t,\mathbb{P}\big(\mathbb{A}(t)\big)\big) =\max_i\big\{A_i(t+D_i)\big\}+\sum_{i=1}^N\mathbf{Q}_i(t) p_{i,t}(\mathbb{A}(t)).\label{cost_V_N}
    \end{equation}
\end{itemize}

Thus, we extent problem (\ref{MDP_formulation}) by considering multiple paths: at the beginning of time slot $t$, 
\begin{equation}
    \begin{aligned}
        &C_t^*\big(\mathbb{S}_t\big) =\min_{\mathbb{P}\big(\mathbb{A}(t)\big)} \sum_{\tau=t}^{T} \rho^{\tau-t} V\Big(\mathbb{S}_t,\mathbb{P}\big(\mathbb{A}(t)\big)\Big),\\
        &\quad \quad\quad \quad\ s.t.\quad (\ref{Est}),(\ref{At_multi})-(\ref{A(i)})\ \text{and}\ (\ref{cost_V_N}).
    \end{aligned}\label{MDP_formulation_M}
\end{equation}
Similar to (\ref{HJB}), (\ref{MDP_formulation_M}) can be rewritten as 
\begin{equation}
\begin{aligned}
   &C_t^*\big(\mathbb{S}_t\big)=\min_{\mathbb{P}(\mathbb{A}(t))}V\Big(\mathbb{S}_t,\mathbb{P}\big(\mathbb{A}(t)\big)\Big) +\rho\mathbb{E}_{\mathbb{S}_{t+1}}\big[C^*_{t+1}\big(\mathbb{S}_{t+1}\big)\big],
\end{aligned}\label{HJB_M}
\end{equation}
where the cost-to-go is
\begin{equation*}
    \begin{aligned}
        \mathbb{E}_{\mathbb{S}_{t+1}}\big[C^*_{t+1}\big(\mathbb{S}_{t+1}\big)\big]=&\mathbf{Q}_a(t) C_{t+1}^*\big(\mathbb{S}_{t+1}^{a}\big)+\mathbf{Q}_b(t) C_{t+1}^*\big(\mathbb{S}_{t+1}^{b}\big)\\&+\sum_{i=1}^N\mathbf{Q}_i(t) C_{t+1}^*\big(\mathbb{S}_{t+1}^{i}\big)
    \end{aligned}
\end{equation*}
according to transition probabilities (\ref{transition_prob_M}).
Note that drivers' random choice model in $q_i(t)$ among $N$ paths incurs totally $(N+2)^{T-\max_i\{D_i\}}$ cost functions to iteratively (\ref{HJB_M}) much more difficult than (\ref{HJB}).

Existing works (e.g.,\cite{sun2018age,kadota2018scheduling,bedewy2019age}) also formulate dynamic program to design optimal scheduling policies for multi-channel network to minimize peak or average weighted AoI. However, the system space of actions only linearly increases in the path number $N$ and time horizon $T$. However, in our problem (\ref{HJB_M}), the system space increases exponentially. In the following, we first reduce multi-path pricing to single-path at a time and then propose our backward-clustered computation to design low-complexity algorithm.

\subsection{Reducing multi-path pricing to single-path at a time}
Note that Propositions \ref{C_increas} and \ref{uniquesol_prop} can be extended for this multi-path model and ensure existence and uniqueness of the online pricing solution to problem (\ref{HJB_M}). Here we need to overcome the curse of dimensionality in both the path number $N$ and time horizon $T$. In this subsection, we prove that we can reduce the $N$ spatial searching dimensions greatly by reducing multi-path pricing to single-path pricing at each time $t$. 
\begin{proposition}\label{m_opt}
To solve problem (\ref{HJB_M}), it is optimal to only price path $\hat{i}_t^*$ out of $N$ paths at each time $t$, where
\begin{equation}
    \hat{i}^*_t=\arg\max_{i\in \{1,\cdots,N\}}\{A_i(t+D_i)\} \label{i_hat_opt}
\end{equation}
\end{proposition}
\begin{proof}
We can first prove that updating path $\hat{i}^*_t$ only is better at any time $t$ by comparing its cost function in (\ref{HJB_M}) with any other single path's cost function. 

Denote the cost functions of pricing single path $m$ and single path $n$ by $C_{m,t}\big(\mathbb{S}(t)\big)$ and $C_{n,t}\big(\mathbb{S}(t)\big)$, respectively, where $m=\hat{i}^*_t$ and $m\neq n$, such that the foreseen AoI of the two paths satisfy
\[
    A_m(t+D_m)>A_n(t+D_n).
\]
Denote their corresponding optimal prices as $p_{m,t}(\mathbb{A}(t))$ and $p_{n,t}(\mathbb{A}(t))$, respectively.

We consider another non-optimal policy
\[
    \hat{p}_{m,t}(\mathbb{A}(t))=\frac{D_m}{D_n}p_{n,t}(\mathbb{A}(t))
\]
for path $m$, then the corresponding non-optimal cost function satisfies
\[
    \hat{C}_{m,t}\big(\mathbb{S}(t)\big)\geq C_{m,t}\big(\mathbb{S}(t)\big).
\]
Note that $p_{n,t}(\mathbb{A}(t)) \in[0,D_n]$, such that the non-optimal price $\hat{p}_{m,t}(\mathbb{A}(t))$ will not exceed $D_m$.

Then we can compare the two cost functions of updating single path $m$ and path $n$ 
\begin{equation}
    \begin{aligned}
        &C_{m,t}\big(\mathbb{S}(t)\big)-C_{n,t}\big(\mathbb{S}(t)\big)\\
        \leq &\hat{C}_{m,t}\big(\mathbb{S}(t)\big)-C_{n,t}\big(\mathbb{S}(t)\big)\\
        =&\rho (1-\mathbf{Q}_n(t))\Big(C_{m,t+1}\big(\mathbb{S}_{t+1}^{b}\big)-C_{n,t+1}\big(\mathbb{S}_{t+1}^{b}\big)\Big)\\ & +\rho \mathbf{Q}_n(t)\Big(C_{m,t+1}\big(\mathbb{S}^{m}_{t+1}\big)-C_{n,t+1}\big(\mathbb{S}^{n}_{t+1}\big)\Big),
    \end{aligned}\label{two_paths_costfunc}
\end{equation}
where $\mathbb{S}^{b}_{t+1}$ and $\mathbb{S}^{i}_{t+1}$ are defined in (\ref{transition_prob_M}). Here we can use the same backward and mathematical induction methods as the proof of Lemma \ref{C_increas} to obtain
\begin{align*}
    &C_{m,t+1}\big(\mathbb{S}^{m}_{t+1}\big)\leq C_{n,t+1}\big(\mathbb{S}^{n}_{t+1}\big)
\end{align*}
and 
\begin{align*}
    &C_{m,t+1}\big(\mathbb{S}^{b}_{t+1}\big)\leq C_{n,t+1}\big(\mathbb{S}^{b}_{t+1}\big).
\end{align*}
Substituting the above two inequalities into (\ref{two_paths_costfunc}), we can finally obtain
\[
    C_{m,t}\big(\mathbb{S}_t\big)<C_{n,t}\big(\mathbb{S}_t\big),
\]
which means that updating the single path $m=\hat{i}^*_t$ (\ref{i_hat_opt}) is better than updating any other single path.

If there are multiple paths with appealing prices for a driver arrival at time $t$ (i.e., in the first case of (\ref{hat_i})), he can only choose one to sample. Then the provider can simultaneously reduce the other non-target paths' prices to zero and properly reduce the target path's price to sample, while keeping the driver's acceptance probability to sample the target path unchanged. This helps save the expected economic payment to the driver and reduce the cost function. Hence, we prove that only pricing path $\hat{i}^*_{t}$ is optimal to solve (\ref{HJB_M}).
\end{proof}

Proposition \ref{m_opt} shows that it is optimal to only price a single path at a time. Intuitively, if positive prices are given to more than one path, the driver will still choose only one to sample. Then the provider need to give higher price for the target path given the competitive prices from others, which incurs unnecessary cost. 

Further, Proposition \ref{m_opt} tells that we should target at the path with the maximum foreseen AoI $\max_i \{A_i(t+D_i)\}$ instead of largest current AoI $\max_i \{A_i(t)\}$. This result is different from \cite{sun2018age} and \cite{kadota2018scheduling}, because the current pricing decision in our problem can only help reduce the AoI after a path-dependent time delay. 

Thanks to Proposition \ref{m_opt}, we only need to design the positive price  $p_{\hat{i}^*_t,t}\big(\mathbb{A}(t)\big)$ for the target path $\hat{i}^*_{t}$ in (\ref{i_hat_opt}) instead of $\mathbb{P}\big(\mathbb{A}(t)\big)$ at each time $t$. Then the cost function in (\ref{cost_V_N}) is simplified to:
\begin{equation}
    V\Big(\mathbb{S}_t,p_{\hat{i}_t^*,t}\big(\mathbb{A}(t)\big)\Big) =A_{\hat{i}_t^*}(t+D_{\hat{i}_t})+\mathbf{Q}_{\hat{i}_t^*}(t)p_{\hat{i}_t^*,t}\big(\mathbb{A}(t)\big).\label{cost_V_N_app}
\end{equation}
Accordingly, (\ref{HJB_M}) can also be simplified as 

\begin{figure}[t]
    \centering
    \includegraphics[width=0.44\textwidth]{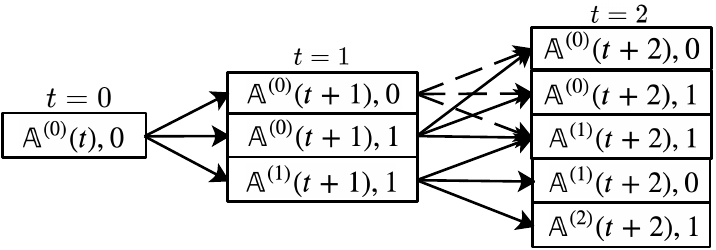}
    \caption{An illustrative example of approximated backward-clustered process with three decision slots ($t=0,1,2$ under $T-\max_i\{D_i\}=2$) in the multi-path scenario.}
    \label{combine_branches}
\end{figure}

\begin{equation}
\begin{aligned}
   &C_t^*\big(\mathbb{S}_t\big)=\min_{p_{\hat{i}_t^*,t}(\mathbb{A}(t))}V\Big(\mathbb{S}_t,p_{\hat{i}_t^*,t}\big(\mathbb{A}(t)\big)\Big) +\rho\mathbb{E}_{\mathbb{S}_{t+1}}\big[C^*_{t+1}\big(\mathbb{S}_{t+1}\big)\big],
\end{aligned}\label{HJB_simp}
\end{equation}
where the cost-to-go is
\begin{equation}
    \begin{aligned}
        \mathbb{E}_{\mathbb{S}_{t+1}}\big[C^*_{t+1}\big(\mathbb{S}_{t+1}\big)\big]=&\mathbf{Q}_a(t) C_{t+1}^*\big(\mathbb{S}_{t+1}^{a}\big)+\mathbf{Q}_b(t) C_{t+1}^*\big(\mathbb{S}_{t+1}^{b}\big)\\&+\mathbf{Q}_{\hat{i}^*_t}(t) C_{t+1}^*\big(\mathbb{S}_{t+1}^{\hat{i}_t^*}\big) \label{S_t+1_M}
    \end{aligned}
\end{equation}
by letting $q_i(t)=0$ for all $i\in\{0,\cdots,N\}$ and $i\neq \hat{i}^*_t$ in (\ref{transition_prob_M}). Though the total number of cost functions has been greatly reduced from $(N+2)^{T-\max_i\{D_i\}}$ to $3^{T-\max_i\{D_i\}}$ in problem (\ref{HJB_simp}), it is still difficult to solve, and the direct application of backward induction approach in Section \ref{section4} cannot help. Here we follow a random pattern to alternate and traverse the target path ($\hat{i}_{t}$) to price over time, and we still need to adapt pricing online to a huge AoI evolution set $\mathbb{A}(t)$ with an exponentially increasing state space in $T$. 

\subsection{Backward-clustered computation for low-complexity approximation algorithm}

In this subsection, we propose a new innovative simplification approach called backward-clustered computation method to solve problem (\ref{HJB_simp}). Different from the repeated long-term cost functions under the same sampled paths we analyzed in Fig. \ref{inductive_D}, the huge number of cost functions across paths are non-repeated themselves. To reduce the exponentially increasing state space, at current time $t$, we cluster those states with the same number of future sampled paths from time $t$ to be one approximated state. By applying this approximation for clustered functions, we no longer count and compute all the long-term cost functions, and we do not need to consider $\hat{i}^*_t$ any more.

More specifically, at current time $t$, we propose to use $\mathbb{A}^{(0)}(t)$ and $C^{(0)}_{t}\big(\mathbb{A}^{(0)}(t),s(t-1)\big)$ to approximate foreseen AoI evolution set $\mathbb{A}(t)$ and its cost function $C^*_t\big(\mathbb{S}_t\big)$ in (\ref{HJB_simp}), respectively, where the superscript $0$ means the number of the sampled paths from time $t$ is still $0$ at time $t$. Note that the initial approximated AoI set $\mathbb{A}^{(0)}(t)$ equals $\mathbb{A}(t)$. Then the three possible states at time $t+1$ in (\ref{transition_prob_M}) can be approximated to $(\mathbb{A}^{(0)}(t+1),0), (\mathbb{A}^{(0)}(t+1),1)$ and $(\mathbb{A}^{(1)}(t+1),1)$, respectively. Here only $\mathbb{A}^{(1)}(t+1)$ with superscript $1$ tells that one path is sampled from time $t$ to $t+1$. Take the clustered approximation of $C^*_t(\mathbb{S}_t)$ in (\ref{HJB_simp}) with three decision slots ($t=0,1,2$ under $T-\max_i\{D_i\}=2$) as an illustrative example, and Fig. \ref{combine_branches} shows the approximated branching network from $t=0$ to $2$ as explained below.

\begin{algorithm}[t]
\caption{Polynomial-time backward-clustered computation of approximation online pricing to (\ref{HJB_simp})}
\label{Approx_multipath}
\begin{algorithmic}[1]
\STATE Observe current AoI and predict maximum foreseen AoI set $\mathbb{A}^{(0)}(t)$;\
\STATE Initialize $\hat{i}^{(0)}_t=\arg\max_i\{A_i^{(0)}(t+D_i)\}$ and $k=T-\max_i\{D_i\}-t$;\
\WHILE{$k \geq 1$} \label{algo3_3}
\STATE Initialize $A_i^{(0)}(t+D_i+k)=A_i^{(0)}(t+D_i)+k$ for all $i\in\{1,\cdots,N\}$;\ \label{algo3_4}
\FOR{$0\leq j\leq k$} \label{algo3_5}
\STATE Obtain $\hat{i}^{(j)}_{t+k}=\arg\max_i \{A_i^{(j)}(t+D_i+k)\}$;\
\STATE Compute $p_{\hat{i}^{(j)}_{t+k},t+k}(\mathbb{A}^{(j)}(t+k))$ using binary search with error $\varepsilon$;\ \label{algo3_7}
\STATE Compute $C^{(j)}_{t+k}\big(\mathbb{A}^{(j)}(t+k),1\big)$ and $C^{(j)}_{t+k}\big(\mathbb{A}^{(j)}(t+k),0\big)$ according to (\ref{HJB_simp}) from last decision slot;\ \label{algo3_8}
\STATE Update $A_{\hat{i}^{(j)}_{t+k}}^{(j+1)}(t+D_{\hat{i}^{(j)}_{t+k}}+k)=D_{\hat{i}^{(j)}_{t+k}}+\frac{k-j}{2}$;\ \label{algo3_9}
\ENDFOR \label{algo3_10}
\STATE $k=k-1$
\ENDWHILE \label{algo3_12}
\STATE Apply binary search of $p_{\hat{i}^{(0)}_t,t}(\mathbb{A}^{(0)}(t))$ with error $\varepsilon$;\ \label{algo3_13}
\RETURN $p_{\hat{i}^{(0)}_t,t}(\mathbb{A}^{(0)}(t))$ for current time $t$
\end{algorithmic}
\end{algorithm}

\begin{itemize}
    \item At $t=2$, there are only $3$ approximate AoI sets (i.e., $\mathbb{A}^{(0)}(t+2),\mathbb{A}^{(1)}(t+2)$ and $\mathbb{A}^{(2)}(t+2)$) in Fig. \ref{combine_branches}, because there are at most two paths being sampled since $t=0$, and we only have $5$ approximated states instead of $3^{T-D}=9$.
    \item At $t=1$, our two approximated states with superscript $0$ (i.e., $\big(\mathbb{A}^{(0)}(1),1\big)$ with last arrival and $\big(\mathbb{A}^{(0)}(1),0\big)$ without last arrival) have the same branching structure to $(\mathbb{A}^{(0)}(2),1)$, $\big(\mathbb{A}^{(0)}(2),0\big)$ and $\big(\mathbb{A}^{(1)}(2),1\big)$ at $t=2$. Thus, it suffices to cluster $\big(\mathbb{A}^{(0)}(1),1\big)$ and $\big(\mathbb{A}^{(0)}(1),0\big)$ by only expanding one of them to $t=2$ in Fig. \ref{combine_branches}.
\end{itemize}

Note that the original states $\mathbb{S}^b_{t+1}$ and $\mathbb{S}^{\hat{i}^*_t}_{t+1}$ in (\ref{S_t+1_M}) cannot branch to the same state at time $t+2$, because the path $\hat{i}^*_t$ is sampled during different time slots, i.e., $t=0$ and $t=1$, respectively. In our approximation of backward-clustered computation, the actual AoI of a sampled path is mildly enlarged to the approximated AoI. Yet this error is always bounded by the elapsed time since current time $t$ and will not take effect until this path is priced to sampled again. By clustering the approximated cost functions with the same number of sampled paths from time $t$ to the whole branching network, we successfully reduce the number of long-term cost functions to compute from $3^{T-\max_i\{D_i\}}$ to around $(T-\max_i\{D_i\})^2$, which can be solved using backward induction.

Regarding the pricing solution, similar to (\ref{u*t}), the approximation online pricing $p_{\hat{i}^{(0)}_t,t}\big(\mathbb{A}^{(0)}(t)\big)$ here satisfies the first-order condition of (\ref{HJB_simp}), yet the cost functions are replaced by their approximation and $\hat{i}^{(0)}_t=\max_i\{A^{(0)}_i(t+D_i)\}$. Now we present Algorithm \ref{Approx_multipath} to efficiently return the approximation online pricing to (\ref{HJB_simp}). We apply our backward-clustered computation to calculate approximation cost functions from $t=T-\max_i{D_i}$ to $t=0$:
\begin{itemize}
    \item From step \ref{algo3_3} to step \ref{algo3_12}, we backward compute all the possible long-term cost functions from time slot $T-\max{D_i}$ to $t+1$, and finally apply binary search to derive $p_{\hat{i}^{(0)}_t,t}(\mathbb{A}^{(0)}(t))$ in step \ref{algo3_13}.
    \item At each future time slot $t+k$, we first calculate the approximated foreseen AoI set $\mathbb{A}^{(j)}(t+k)$ for all $j\in\{0,\cdots,k\}$. At $j=0$, all the paths are not sampled, such that $\mathbb{A}^{(0)}(t+k)$ is initialized at step \ref{algo3_4}. Then we choose to update the AoI of path $\hat{i}^{(j)}_{t+k}$ to increase $j$ to $j+1$ at $j$-th loop in step \ref{algo3_9}.
    \item In steps \ref{algo3_7} and \ref{algo3_8}, we compute the corresponding approximation price and cost functions of each AoI set $\mathbb{A}^{(j)}(t+k)$.
\end{itemize}
Note that in step \ref{algo3_9}, we update the AoI of path $\hat{i}_{t+k}^{j}$ to $D_{\hat{i}_{t+k}^{j}}+\frac{k-j}{2}$ because this $j$-th updated path can be sampled at any time from $t+j$ to $t+k$, such that we take the average of all the possible AoI based on our linear model in (\ref{cost_V_N}).

As a counter-part of (\ref{HJB_simp}), we denote $\hat{C}_t\big(\mathbb{A}(t),s(t-1)\big)$ as the approximated cost function under our approximation pricing, and we next prove that it has bounded performance gap to the optimum in the following theorem.
\begin{theorem}\label{Multi_approx_thm}
At any time $t$, our approximation Algorithm \ref{Approx_multipath} takes only polynomial time complexity $O\big((T-\max_i\{D_i\})^2\log_2\big(\frac{\max_i\{D_i\}}{\varepsilon}\big)\big)$ to return the approximation online pricing $p_{\hat{i}^{(0)}_t,t}\big(\mathbb{A}^{(0)}(t)\big)$, and the complexity does not depend on the path number $N$. As compared to the cost optimal objective $C_t^*\big(\mathbb{A}(t)),s(t-1)\big)$ in (\ref{HJB_simp}), Algorithm \ref{Approx_multipath}'s cost objective $\hat{C}_t\big(\mathbb{A}(t)),s(t-1)\big)$ achieves the following approximation error in the worst case:
\begin{equation}
    \begin{aligned}
    &\big|\hat{C}_t\big(\mathbb{A}(t)),s(t-1)\big)-C_t^*\big(\mathbb{A}(t)),s(t-1)\big)\big|\\
    <&\frac{\rho^{g(N)+1}-\rho^{T-\max_i\{D_i\}-t+1}}{(1-\rho)^2},\label{performance_bound}
    \end{aligned}
\end{equation}
and this error increases with $T$ but decreases with $N$.
\end{theorem}

\begin{proof}
As illustrated in Fig. \ref{combine_branches}, our approximation of the clustered computation reduces the number of cost functions from $3^{T-\max_i\{D_i\}}$ to only around $(T-\max_i\{D_i\})^2$ for each time $t$. It takes at most complexity $\log_2\big(\frac{\max_i\{D_i\}}{\varepsilon}\big)$ to obtain the approximation pricing for each approximated cost function, by applying binary search of the first-order condition of (\ref{HJB_simp}) with error $\varepsilon$. Thus, it takes complexity $O\big((T-\max_i\{D_i\})^2\log_2\big(\frac{\max_i\{D_i\}}{\varepsilon}\big)\big)$ to obtain the approximation pricing for any particular time $t$. 

Since the actual AoI set is enlarged to the approximated AoI set, the approximated cost function satisfies 
\[
    C_t^{(0)}\big(\mathbb{A}^{(0)}(t),s(t-1)\big)>\hat{C}_t\big(\mathbb{A}(t)),s(t-1)\big)
\]
Then we can use such relationship to compute the approximation error:
\begin{equation}\notag
    \begin{aligned}
        &\hat{C}_t\big(\mathbb{A}(t)),s(t-1)\big)-C_t^*\big(\mathbb{A}(t)),s(t-1)\big)\\
        \leq &C_t^{(0)}\big(\mathbb{A}^{(0)}(t),s(t-1)\big)-C_t^*\big(\mathbb{A}(t)),s(t-1)\big)\\
        \leq& \rho\Big(C_{t+1}^{(1)}\big(\mathbb{A}^{(1)}(t+1),1\big)-C_{t+1}^*\big(\mathbb{A}(t+1,\hat{i}^*_t),1\big)\Big)\\
        \leq& \rho \sum_{n=g(N)}^{T-\max_i\{D_i\}-t}\rho^n n
        \\ =& \frac{\rho^{g(N)+1}-\rho^{T-\max_i\{D_i\}-t+1}}{(1-\rho)^2}.
    \end{aligned}
\end{equation}

\begin{figure}[t]
    \centering
    \subfigure[Evolution of foreseen AoI in all the three paths under the approximation online pricing $p_{\hat{i}_t^{(0)},t}(\mathbb{A}^{(0)}(t)) , t\in \{0,\cdots,T-\max_i{D_i}=25\}$.]{
    \includegraphics[width=0.4\textwidth]{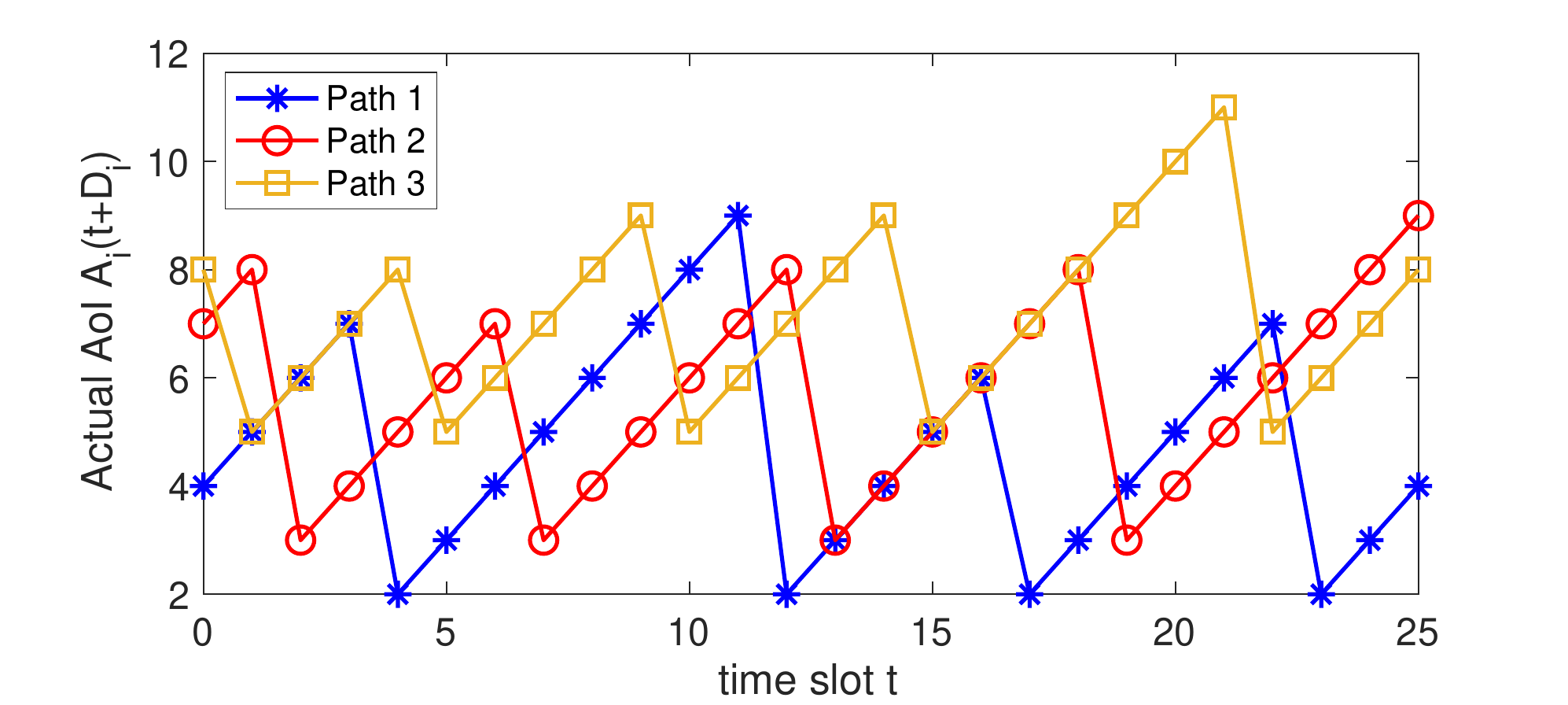}\label{multi_age}}
    \subfigure[Online approximation pricing policy returned by Algorithm \ref{Approx_multipath}.]{
    \includegraphics[width=0.41\textwidth]{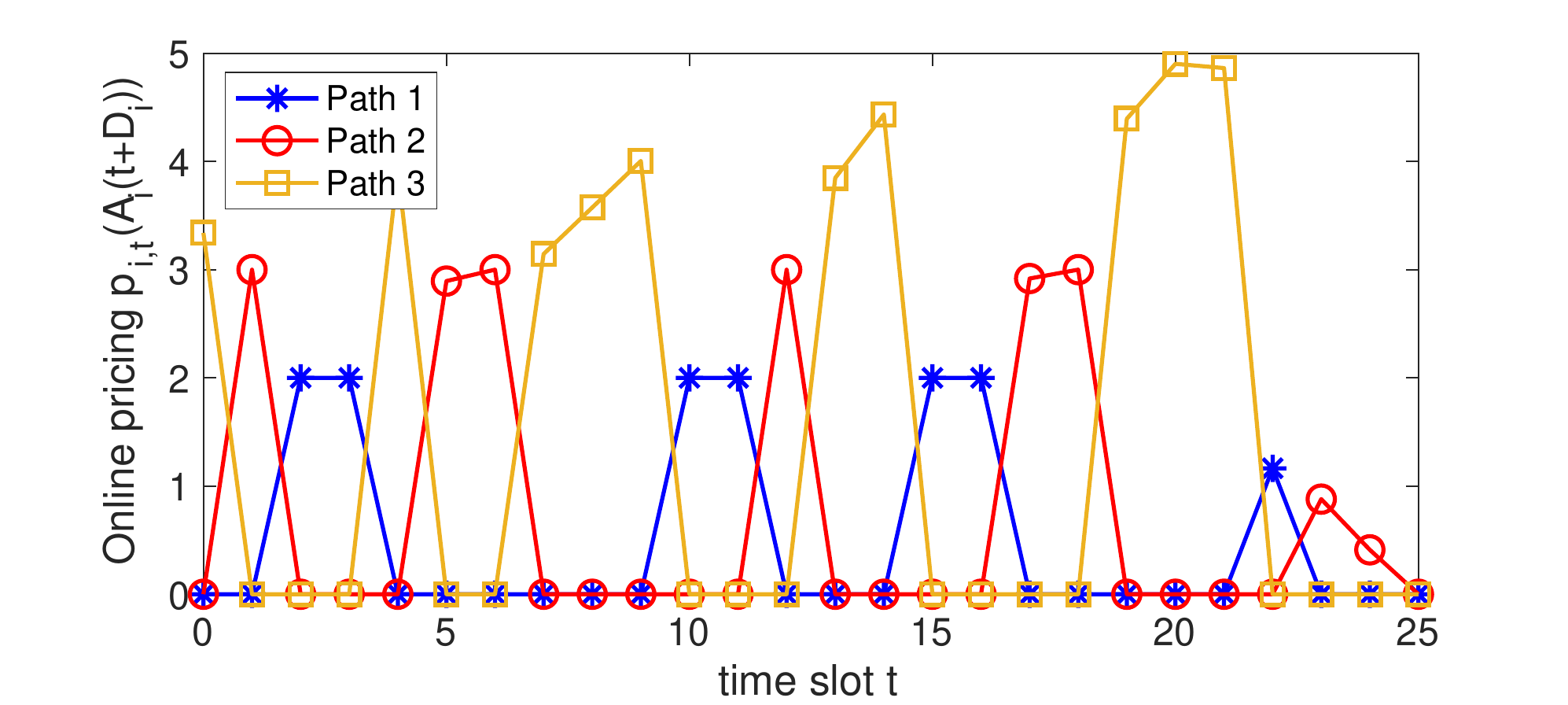}\label{multi_price}}
    \caption{The evolution of foreseen AoI set $\mathbb{A}(t)$ and online approximation price $p_{\hat{i}_t^{(0)},t}(\mathbb{A}^{(0)}(t))$ returned by Algorithm \ref{Approx_multipath} from $t=0$ to $25$. Here we set $\mathbb{A}(0)=\{2,4,3\}, D=\{2,3,5\}$ with $N=3$, $T=30, \rho=0.85$, probabilities $\alpha=0.8,\beta=0.6$ for the Markov chain in Fig. \ref{event}, and the cost CDF $F(x)$ satisfies a truncated normal distribution with mean $0.6$ and variance $0.7$.}
    \label{multi_pricing_fig}
\end{figure}

The second inequality above is obtained using the first-order conditions of the original and approximated cost functions with respect to their prices at time $t$, and the third inequality is because the performance loss at any future time $t+n$ cannot exceed the elapsed time $n$. The function $g(N)$ above is an increasing function of $N$ to return the updating time cycle to sample path $\hat{i}^*_t$ again.
We can check the partial derivatives of the right-hand-side bound of (\ref{performance_bound}) to show that it increases with $T$ but decreases with $N$.
\end{proof}

The complexity of our Algorithm \ref{Approx_multipath} does not depend on path number $N$ and can thus fit large-scale road networks. Actually, the approximation error on the right hand side of (\ref{performance_bound}) is small most of the time. Perhaps surprisingly, it decreases with a larger path number $N$. This is because individual paths become similar to each other in a larger path choice pool, and our backward-clustered computation in Algorithm \ref{Approx_multipath} only counts the number of sampled paths (without checking which paths) for approximating cost functions.   

We create a typical instance but with $N=3$ paths in Fig. \ref{multi_pricing_fig} to run the online approximation pricing $p_{\hat{i}_t^{(0)},t}(\mathbb{A}^{(0)}(t))$ and corresponding foreseen AoI set $\mathbb{A}(t)$, $\forall t\in\{0, \cdots, T=25\}$ returned by Algorithm \ref{Approx_multipath}. Here we set the travel delay set $D=\{2,3,5\}$ for the three paths, whose initial AoI $\mathbb{A}(0)=\{2,4,3\}$, and the cost distributions of each driver to follow a truncated normal distribution with mean $0.6$ and variance $0.7$. 

Fig. \ref{multi_age} shows the evolution of the foreseen AoI of all the three paths, and the maximum AoI among all the paths are well controlled thanks to our online pricing solution in Fig. \ref{multi_price}. We can also see that the platform will select the path with the maximum AoI to price, which is consistent with Proposition \ref{m_opt}. In this case, the online pricing follows a $D_i$-delayed pattern of the actual AoI over time.

\begin{figure}[t]
    \centering
    \subfigure[Actual performance error of truncated normal distribution with mean $0.6$ and variance $0.7$.]{
    \includegraphics[width=0.38\textwidth]{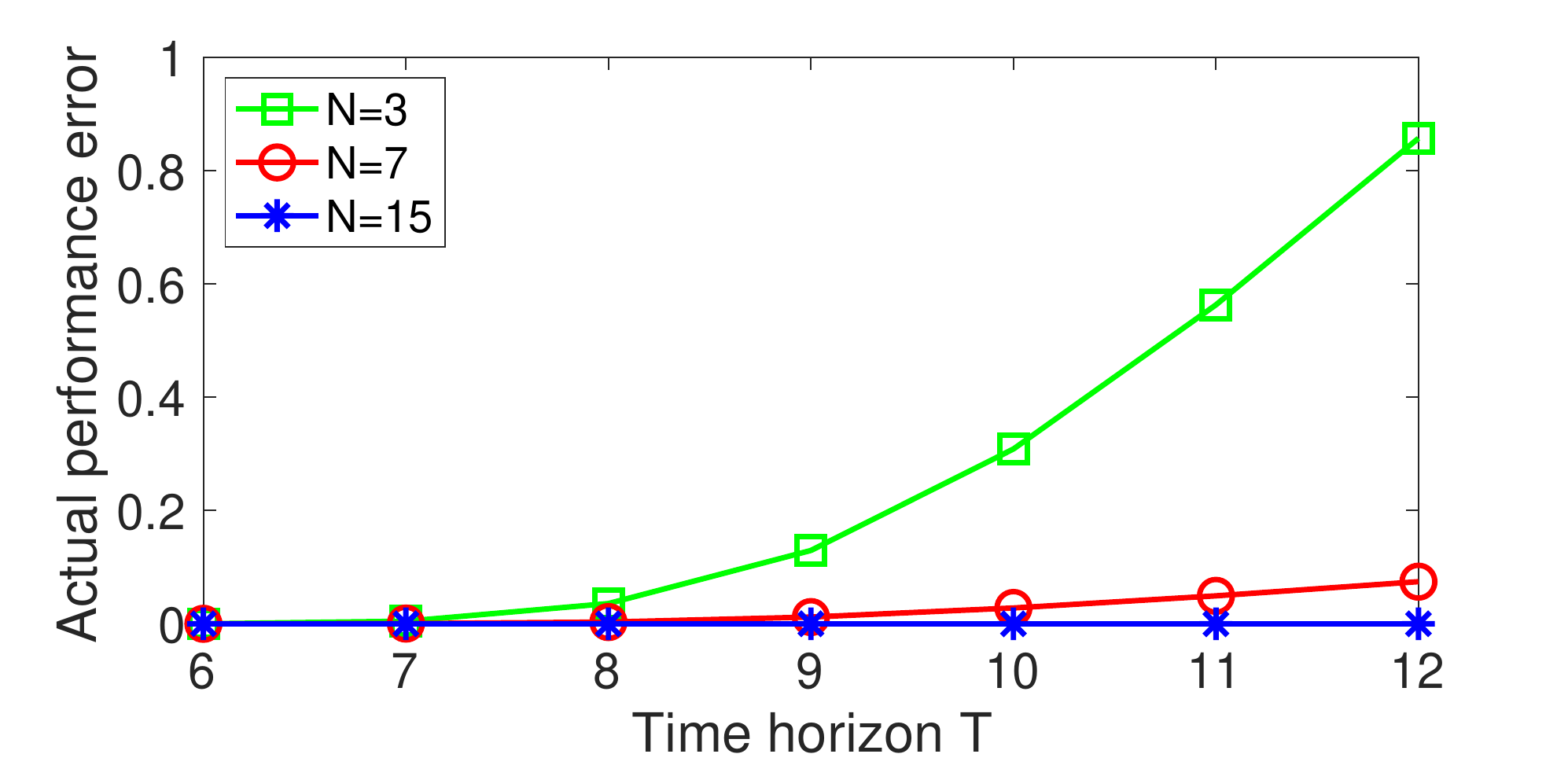}\label{approx_error_norm}}
    \subfigure[Actual performance error of truncated logistic distribution with location parameter 0.6 and scale 1.]{
    \includegraphics[width=0.38\textwidth]{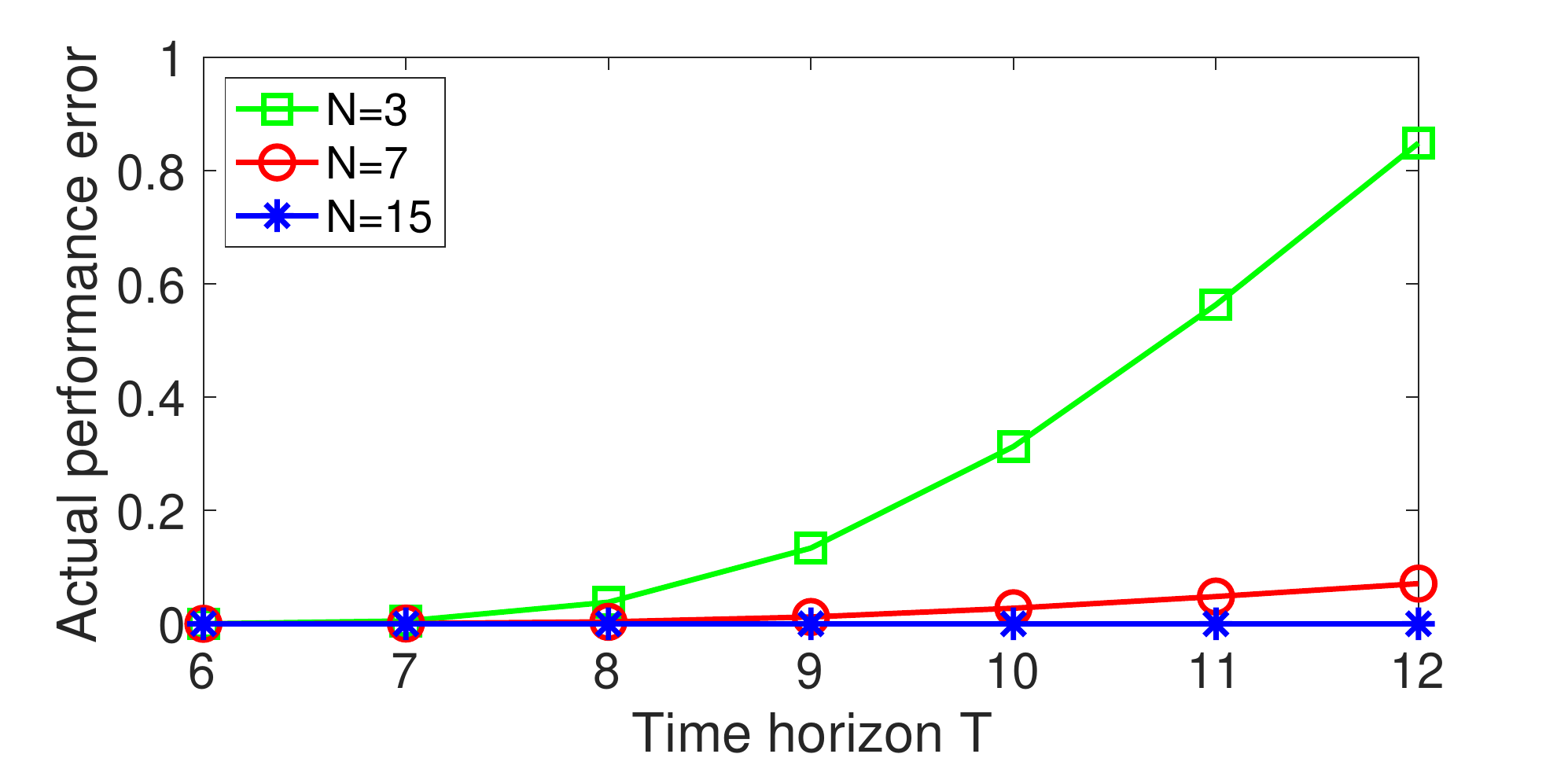}\label{approx_error_logi}}
    \caption{The actual performance loss of our Algorithm \ref{Approx_multipath} versus path number $N$ and time horizon $T$ under truncated normal and truncated logistic distributions. Here we vary the path numbers $N=3,7$ and $15$, time horizon $T=\{6,\cdots,12\}$, and consider the time $t=0$'s online decision making and set maximum travel delay among all paths as $\max_i\{D_i\}=5$.}
    \label{approx_error}
\end{figure}

Finally, unlike (\ref{performance_bound})'s error for the worst-case, we run simulations to examine Algorithm \ref{Approx_multipath}'s actual performance loss as compared to the optimum, versus time horizon $T$ and path number $N$ in Fig. \ref{approx_error}. Here we generally consider initial $t=0$'s online decision making, then the cost objective $\hat{C}_0\big(\mathbb{A}(0)),0\big)$ uses the approximation pricing returned by Algorithm \ref{Approx_multipath}. While the optimal cost function $C_0^*\big(\mathbb{A}(0)),0\big)$ in (\ref{HJB_simp}) uses the optimal pricing returned by traditional iterative computation methods, whose complexity is high and only allows us to keep $T\leq 12$ for simulation. In Fig. \ref{approx_error} with almost the same Fig. \ref{approx_error_norm} and Fig. \ref{approx_error_logi}, we find that effect of $F(x)$ to the performance error is limited, which shows our scalable Algorithm \ref{Approx_multipath}'s robustness to different cost distributions. The performance error is mild from the optimum under our algorithm. This result is also consistent with Theorem \ref{Multi_approx_thm} to show the approximation error increase in time horizon $T$ and decrease in path number $N$.

\section{Conclusion}\label{section6}
In this paper, we have proposed online pricing for a content provider to economically reward drivers for diverse routing and control the actual AoI dynamics over time and spatial path domains. This online pricing optimization problem should be solved without knowing drivers' costs and even arrivals, and is intractable due to the curse of dimensionality in both time and space. If there is only one non-shortest path, we leverage the Markov decision process (MDP) techniques to analyze the problem. Accordingly, we design a linear-time algorithm for returning optimal online pricing, where a higher pricing reward is needed for a larger AoI. We also show that our algorithm is applicable to infinite time horizon. If there are a number of non-shortest paths, we prove that pricing one path at a time is optimal, yet it is not optimal to choose the path with the largest current AoI. Then we propose a new backward-clustered computation method and develop an approximation algorithm to alternate different paths to price over time. This algorithm has only polynomial-time complexity, and its complexity does not depend on the number of paths. Perhaps surprisingly, our analysis of approximation ratio suggests that our algorithm's performance approaches closer to optimum if more paths are involved to sample.

We can consider some possible future works directions to extend this work. Recall that this paper focuses on the min-max foreseen AoI optimization in problem (\ref{HJB_M}) if there are a number of non-shortest paths. We can extend to study the min-sum AoI optimization problem to minimize the total AoI across all the paths. Our backward-clustered computation method should still work and we should still choose to pricing one path at a time. Yet we may not choose the path with the largest foreseen AoI in Proposition \ref{m_opt}.
Moreover, we can also extend our analysis to deal with the random travel delay instead of deterministic delay in each path, where we need to further take expectation in the online optimization. The evaluation about the performance of our algorithms with real data is another point in our next step.


%



\ifCLASSOPTIONcaptionsoff
  \newpage
\fi



%

%


\begin{IEEEbiography}[{\includegraphics[width=1in,height=1.25in,clip,keepaspectratio]{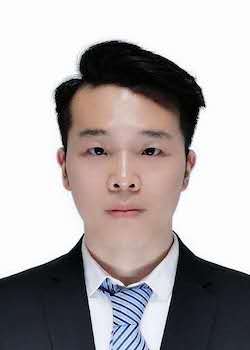}}]{Hongbo Li}
received the B.S. degree in Electronics and Electric Engineering from Shanghai Jiao Tong University, Shanghai, China, in 2019. He is currently working toward the Ph.D. degree with the Pillar of Engineering Systems and Design, Singapore University of Technology and Design (SUTD). His research interests include network economics, game theory, and mechanism design.
\end{IEEEbiography}

\begin{IEEEbiography}[{\includegraphics[width=1in,height=1.25in,clip,keepaspectratio]{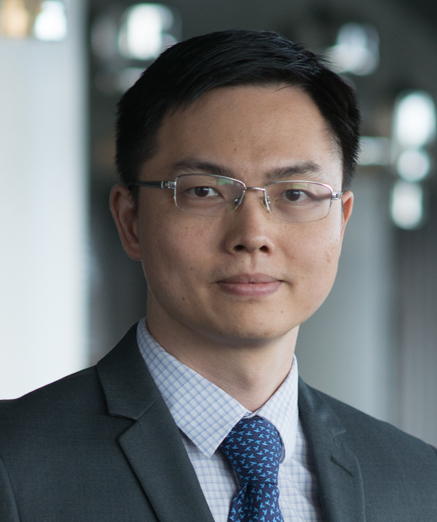}}]{Lingjie Duan}(S’09-M’12-SM’17)
received the Ph.D. degree from The Chinese University of Hong Kong in 2012. He is an Associate Professor of Engineering Systems and Design with the Singapore University of Technology and Design (SUTD). In 2011, he was a Visiting Scholar at University of California at Berkeley, Berkeley,CA, USA. His research interests include network economics and game theory, cognitive and green networks, and energy harvesting wireless communications. He is an Editor of IEEE Transactions on Wireless Communications. He was an Editor of IEEE Communications Surveys and Tutorials. He also served as a Guest Editor of the IEEE Journal on Selected Areas in Communications Special Issue on Human-in-the-Loop Mobile Networks, as well as IEEE Wireless Communications Magazine. He received the SUTD Excellence in Research Award in 2016 and the 10th IEEE ComSoc Asia-Pacific Outstanding Young Researcher Award in 2015.
\end{IEEEbiography}





\end{document}